%% file: tr.tex
\documentclass[a4paper]{llncs}
\usepackage{fullpage}
\usepackage{microtype}
\usepackage{amssymb}
\usepackage{amsmath}
\usepackage{graphicx}
\usepackage{url}
\usepackage{breakurl}
\usepackage{float}
\usepackage{manfnt}
\usepackage{listings}
\usepackage{color}
\usepackage{mathpartir}
\usepackage{hyperref}
\usepackage{xspace}
\usepackage{pifont}

\newcommand{\includeapp}{}

\newcommand*{\appref}[1]{%
\ifdefined\includeapp
Appendix~\ref{app:#1}%
\else%
the appendix of the extended paper \cite{fu16}%
\fi%
}

\input{macros}
\title{\muPuppet: A Declarative Subset of the Puppet Configuration Language}

\author{Weili Fu \and Roly Perera \and Paul Anderson \and James Cheney}
\institute{Laboratory for Foundations of Computer Science, School of Informatics, University of Edinburgh}
\date{}

\begin{document}

\maketitle

\input{abstract}
\input{sec/introduction}
\input{sec/overview}

\input{sec/semantics}
\input{sec/metatheory}
\input{sec/implementation}
\input{sec/related-work}

\input{sec/conclusion}

\input{acks.tex}
\bibliography{paper}
\appendix
\input{appendix}
\end{document}

%% file: macros.tex
\newcommand*{\etal}{et al.\@\xspace}

\newcommand*{\bnf}{\mathrel{::=}}

\newcommand*{\nodeMatch}{\mathsf{nodeMatch}}

\newcommand*{\caseMatch}{\mathsf{caseMatch}}
\newcommand*{\ClassDef}{\mathsf{ClassDef}}
\newcommand*{\ResourceDef}{\mathsf{ResourceDef}}
\newcommand*{\Declared}{\mathsf{DeclaredClass}}

\newcommand*{\RV}{v_\mathrm{R}}
\newcommand*{\HV}{v_\mathrm{H}}
\newcommand*{\CV}{v_\mathrm{C}}

\newcommand*{\bang}{\mathtt{!}}
\newcommand*{\query}{\mathtt{?}}

\newcommand*{\dom}{\mathit{dom}}

\newcommand*{\Set}[1]{\mathit{#1}}
\newcommand*{\baseof}[1]{\mathbin{\mathsf{baseof}_{#1}}}
\newcommand*{\parentof}[1]{\mathbin{\mathsf{parentof}_{#1}}}
\newcommand*{\stepExpr}[1]{\xrightarrow{#1}}
\newcommand*{\stepStmt}[1]{\xrightarrow{#1}_{\mathsf{s}}}
\newcommand*{\stepManifest}[1]{\xrightarrow{#1}_{\mathsf{m}}}
\newcommand*{\stepHash}[1]{\xrightarrow{#1}_{\mathsf{H}}}
\newcommand*{\stepRes}[1]{\xrightarrow{#1}_{\mathsf{R}}}

\newcommand*{\visiblespace}{\text{\textvisiblespace}}
\newcommand*{\seq}{\visiblespace}

\newcommand*{\clear}{\mathit{clear}}
\newcommand*{\lookup}{\mathit{lookup}}
\newcommand*{\lookupCat}{\mathit{lookup_\mathrm{C}}}
\newcommand*{\merge}{\mathit{merge}}
\newcommand*{\Not}[1]{\ensuremath{{\bang}\,#1}}

\newcommand{\Paragraph}[1]{\textbf{#1}.}

\newcommand{\Yes}{\checkmark}
\newcommand{\No}{\ding{55}}

\newcommand*{\kw}[1]{{\tt #1}\xspace}

\newcommand*{\kcase}{\kw{case}}
\newcommand*{\kclass}{\kw{class}}
\newcommand*{\kdefine}{\kw{define}}
\newcommand*{\kdefault}{\kw{default}}
\newcommand*{\kelse}{\kw{else}}
\newcommand*{\kelseif}{\kw{elsif}}
\newcommand*{\kfalse}{\kw{false}}
\newcommand*{\kif}{\kw{if}}
\newcommand*{\kinclude}{\kw{include}}
\newcommand*{\kinherits}{\kw{inherits}}
\newcommand*{\knode}{\kw{node}}
\newcommand*{\kscope}{\kw{scope}}
\newcommand*{\kskip}{\kw{skip}}
\newcommand*{\ktrue}{\kw{true}}
\newcommand*{\kunless}{\kw{unless}}

\newcommand*{\kand}{\mathbin{\kw{and}}}
\newcommand*{\kor}{\mathbin{\kw{or}}}

\newcommand*{\scopeSep}{{::}}
\newcommand*{\topScope}{{::}}
\newcommand*{\classScope}[1]{{::}#1}
\newcommand*{\nodeScope}{{::}\kw{nd}}
\newcommand*{\defScope}[1]{#1\;\kw{def}}

\newcommand*{\fupdate}[4]{#1[(#2,#3):#4]}
\newcommand*{\defupdate}[3]{#1[#2:#3]}

\newcommand*{\muPuppet}{$\mu$Puppet\xspace}

\definecolor{bgcolor}{rgb}{0.97,0.97,0.97}
\definecolor{gray}{rgb}{0.5,0.5,0.5}
\newcommand*{\plab}[1]{\label{#1}}

\urldef{\mailsa}\path|{,|
\urldef{\mailsb}\path|,|
\urldef{\mailsc}\path|,}@ed.ac.uk|

\newcommand{\puppetversion}{4.8\xspace}

%% file: abstract.tex
\begin{abstract}
  Puppet is a popular declarative framework for specifying and managing
  complex system configurations. The Puppet framework includes a
  domain-specific language with several advanced features inspired by
  object-oriented programming, including user-defined resource types,
  `classes' with a form of inheritance, and dependency management. Like most
  real-world languages, the language has evolved in an ad hoc fashion,
  resulting in a design with numerous features, some of which are complex,
  hard to understand, and difficult to use correctly.

  We present an operational semantics for \muPuppet, a representative subset of the
  Puppet language that covers the distinctive features of Puppet, while
  excluding features that are either deprecated or work-in-progress.
  Formalising the semantics sheds light on difficult parts of the language,
  identifies opportunities for future improvements, and provides a foundation
  for future analysis or debugging techniques, such as static typechecking or
  provenance tracking. Our semantics leads straightforwardly to a reference
  implementation in Haskell. We also discuss some of Puppet's idiosyncrasies,
  particularly its handling of classes and scope, and present an initial
  corpus of test cases supported by our formal semantics.
\end{abstract}

%% file: sec/introduction.tex
\section{Introduction}

Managing a large-scale data center consisting of hundreds or thousands of
machines is a major challenge. Manual installation and configuration is simply
impractical, given that each machine hosts numerous software components, such
as databases, web servers, and middleware. Hand-coded configuration scripts
are difficult to manage and debug when multiple target configurations are
needed. Moreover, misconfigurations can potentially affect millions of users.
Recent empirical studies~\cite{yin11sosp,gunawi14socc} attribute a significant
proportion of system failures to misconfiguration rather than bugs in the
software itself. Thus better support for specifying, debugging and verifying
software configurations is essential to future improvements in
reliability~\cite{xu15cs}.

A variety of \emph{configuration frameworks} have been developed to increase
the level of automation and reliability. All lie somewhere on the spectrum
between ``imperative'' and ``declarative''. At the imperative end, developers
use conventional scripting languages to automate common tasks. It is left to
the developer to make sure that steps are performed in the right order, and
that any unnecessary tasks are not (potentially harmfully) executed anyway. At
the declarative end of the spectrum, the desired system configuration is
\emph{specified} in some higher-level way and it is up to the configuration
framework to determine how to \emph{realise} the specification: that is, how
to generate a compliant configuration, or adapt an already-configured
system to match a new desired specification.

Most existing frameworks have both imperative and declarative aspects.
Chef~\cite{marschall13}, CFEngine~\cite{zamboni12}, and
Ansible~\cite{geerling15} are imperative in relation to dependency management;
the order in which tasks are run must be specified. Chef and CFEngine are
declarative in that a configuration is specified as a desired target state,
and only the actions necessary to end up in a compliant state are executed.
(This is called \emph{convergence} in configuration management speak.) The
Puppet framework~\cite{puppet} lies more towards the declarative end, in
that the order in which configuration tasks are carried out is also left
mostly to the framework. Puppet also provides a self-contained
\emph{configuration language} in which specifications are written, in contrast
to some other systems. (Chef specifications are written in Ruby, for example,
whereas Ansible is YAML-based.)

Configuration languages often have features in common with general-purpose
programming languages, such as variables, expressions, assignment, and
conditionals. Some, including Puppet, also include ``object-oriented''
features such as classes and inheritance. However, (declarative) configuration
languages differ from regular programming or scripting languages in that they
mainly provide mechanisms for specifying, rather than realising,
configurations. While some ``imperative'' features that can directly mutate
system state are available in Puppet, their use is generally discouraged.

Like most real-world languages, configuration languages have largely evolved
in an ad hoc fashion, with little attention paid to their semantics. Given
their infrastructural significance, this makes them an important (although
challenging) target for formal study: a formal model can clarify difficult or
counterintuitive aspects of the language, identify opportunities for
improvements and bug-fixes, and provide a foundation for static or dynamic
analysis techniques, such as typechecking, provenance tracking and execution
monitoring. In this paper, we investigate the semantics of the configuration
language used by the Puppet framework. Puppet is a natural choice because of
its DSL-based approach, and the fact that it has seen widespread adoption. The
2016 PuppetConf conference attracted over 1700 Puppet users and developers and
sponsorship from over 30 companies, including Cisco, Dell,
Microsoft, Google, Amazon, RedHat, VMWare, and Citrix.

An additional challenge for the formalisation of real-world languages is that
they tend to be moving targets. For example, Puppet 4.0, released in March
2015, introduced several changes that are not backwards-compatible with Puppet
3, along with a number of non-trivial new features. In this paper, we take
Puppet \puppetversion (the version included with Puppet Enterprise 2016.5) as
the baseline version of the language, and define a subset called \muPuppet
(pronounced ``muppet'') that includes the established features of the language
that appear most important and distinctive; in particular, it includes the
constructs \verb|node|, \verb|class|, and \verb|define|. These are used in
almost all Puppet programs (called
\emph{manifests}). We chose to exclude some features that are either
deprecated or not yet in widespread use, or whose formalisation would add
complication without being particularly enlightening, such as regular
expressions and string interpolation.

The main contributions of this paper are:

\begin{enumerate}
\item a formalisation of \muPuppet, a subset of Puppet \puppetversion;
\item a discussion of simple metatheoretic properties of \muPuppet
  such as determinism, monotonicity and (non-)termination;
\item a reference implementation of \muPuppet in Haskell;
\item a corpus of test cases accepted by our implementation;
\item a discussion of the more complex features not handled by \muPuppet.
\end{enumerate}

We first give an overview of the language via some examples
(Section~\ref{sec:overview}), covering some of the more
counterintuitive and surprising parts of the language. Next we define
the abstract syntax and a small-step operational semantics of
\muPuppet (Section~\ref{sec:language}). We believe ours to be the
first formal semantics a representative subset of Puppet; although
recent work by Shambaugh et al.~\cite{shambaugh16pldi} handles some
features of Puppet, they focus on analysis of the ``realisation''
phase and do not present a semantics for the \verb|node| or
\verb|class| constructs or for inheritance (although their
implementation does handle some of these features).  We use a small-step
operational semantics (as opposed to large-step or denotational semantics)
because it is better suited to modelling some of the idiosyncratic aspects of
Puppet, particularly the sensitivity of scoping to evaluation order. We focus
on unusual or novel aspects of the language in the main body of the paper; the
full set of rules are given in \appref{rules}. Section~\ref{sec:metatheory}
discusses some properties of \muPuppet, such as determinism and monotonicity,
that justify calling it a `declarative' subset of Puppet.
Section~\ref{sec:impl} describes our implementation and how we validated our
rules against the actual behaviour of Puppet, and discusses some of the omitted
features. Sections~\ref{sec:related} and~\ref{sec:concl} discuss related work
and present our conclusions.

%% file: sec/overview.tex
\section{Overview of Puppet}
\label{sec:overview}

\begin{figure}[tb]
\centering
  \includegraphics[scale=0.5]{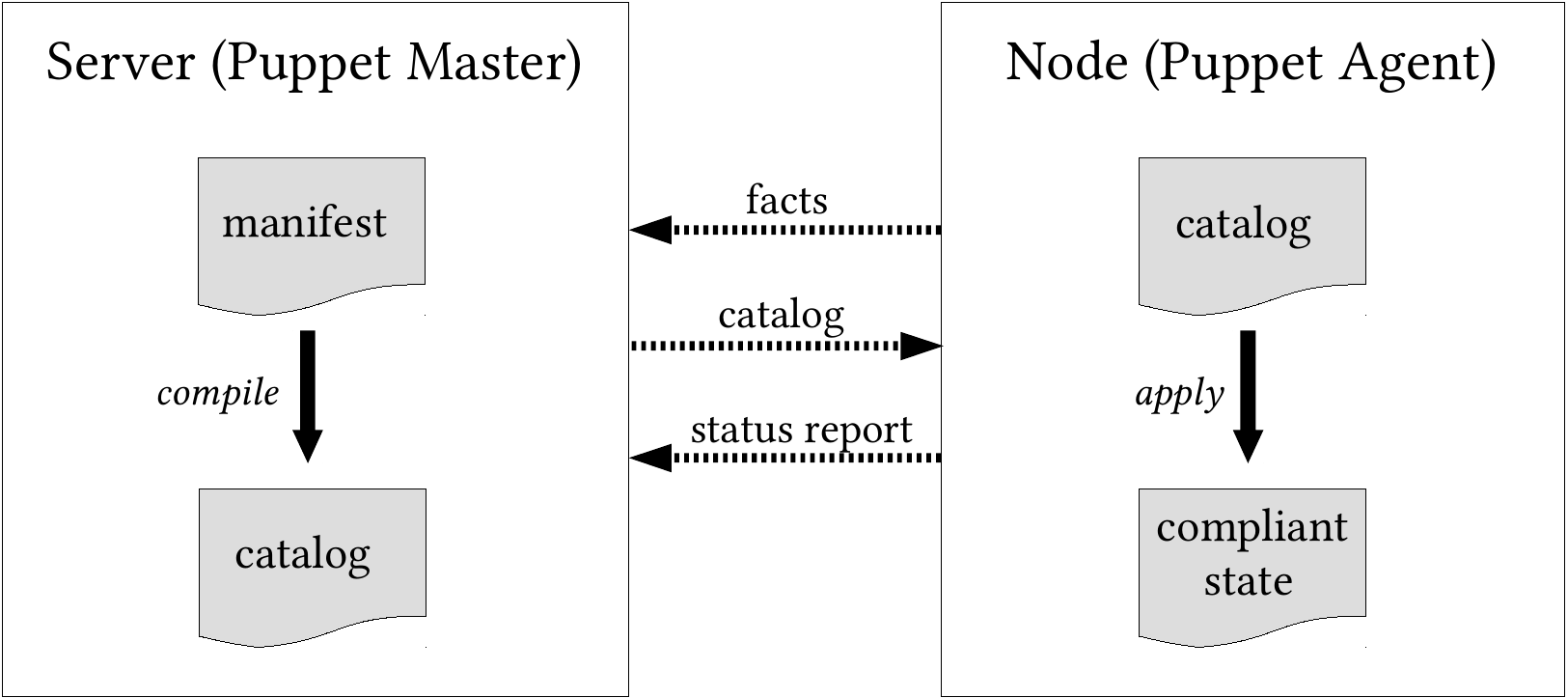}
  \caption{Puppet overview}
  \label{fig:overview}
\end{figure}

Puppet uses several terms -- especially \emph{compile}, \emph{declare}, and
\emph{class} -- in ways that differ from standard usage in programming
languages and semantics.  We introduce these terms with their Puppet meanings
in this section, and use those meanings for the rest of the paper. To aid the
reader, we include a glossary of Puppet terms in \appref{glossary}.

The basic workflow for configuring a single machine (\emph{node})
using Puppet is shown in Figure~\ref{fig:overview}.  A \emph{Puppet
  agent} running on the node to be configured contacts the
\emph{Puppet master} running on a server, and sends a check-in request
containing local information, technically called \emph{facts}, such as
the name of the operating system running on the client node. Using
this information, along with a centrally maintained configuration
specification called the \emph{manifest}, the Puppet master
\emph{compiles} a \emph{catalog} specific to that node. The manifest
is written in a high-level language, the Puppet programming language
(often referred to simply as Puppet), and consists of
\emph{declarations} of \emph{resources}, along with other program
constructs used to define resources and specify how they are assigned
to nodes. A resource is simply a collection of key-value pairs, along
with a \emph{title}, of a particular \emph{resource type};
``declaring'' a resource means specifying that a resource of that type
exists in the target configuration. The catalog resulting from
compilation is the set of resources computed for the target node,
along with other metadata such as ordering information among
resources. The Puppet master may fail to compile a manifest due to
compilation errors.  In this case, it will not produce a compiled
catalog.  If compilation succeeds, the agent receives the compiled
catalog and \emph{applies} it to reconfigure the client machine,
ideally producing a compliant state. Puppet separates the compilation
of manifests and the deployment of catalogs. After deploying the
catalog, either the changed configuration meets the desired
configuration or there are some errors in it that cause system
failures. Finally, the agent sends a status report back to the master
indicating success or failure.

Figure~\ref{fig:overview} depicts the interaction between a single agent and
master. In a large-scale system, there may be hundreds or thousands of nodes
configured by a single master. The manifest can describe how to configure all
of the machines in the system, and parameters that need to be coordinated
among machines can be specified in one place. A given run of the Puppet
manifest compiler considers only a single node at a time.

\subsection{Puppet: key concepts}

We now introduce the basic concepts of the Puppet language -- manifests,
catalogs, resources, and classes -- with reference to various examples. We
also discuss some behaviours which may seem surprising or unintuitive;
clarifying such issues is one reason for pursuing a formal definition of the
language. The full Puppet
\puppetversion language has many more features than presented here.
A complete list of features and the subset supported by \muPuppet are given in
\appref{features}.

\subsubsection{Manifests and catalogs}

Figure~\ref{ex1} shows a typical manifest, consisting of a \emph{node
definition} and various \emph{classes} declaring resources, which will be
explained in \S~\ref{sec:overview:classes} below. Node definitions, such as
the one starting on line~\ref{ex1:node}, specify how a single machine or group
of machines should be configured. Single machines can be specified by giving a
single hostname, and groups of machines by giving a list of hostnames, a
regular expression, or
\kdefault (as in this example). The \kdefault node definition is used if no
other definition applies.

In this case the only node definition is \kdefault, and so compiling this
manifest for any node results in the catalog on the right of Figure~\ref{ex1}.
In this case the catalog is a set of resources of type \kw{file} with titles
\kw{config1},
\kw{config2} and \kw{config3}, each with a collection of attribute-value
pairs. Puppet supports several persistence formats for catalogs, including
YAML; here we present the catalog using an abstract syntax which is
essentially a sub-language of the language of manifests. The
\kw{file} resource type is one of Puppet's many built-in resource types, which
include other common configuration management concepts such as
\verb|user|, \verb|service| and \verb|package|.

\subsubsection{Resource declarations}
\label{sec:overview:resources}

Line~\ref{ex1:config1} of the manifest in Figure~\ref{ex1} shows how the
\kw{config1} resource in the catalog was originally declared. The \kw{path}
attribute was specified explicitly as a string literal; the other attributes
were given as variable references of the form
\kw{\$$x$}. Since a resource with a given title and type is global to the
entire catalog, it may be declared only once during a given compilation. A
\emph{compilation error} results if a given resource is declared more than
once. Note that what Puppet calls a ``compilation error'' is a purely dynamic
condition, and so is really a runtime error in conventional terms.

The ordering of attributes within a resource is not significant; by default
they appear in the catalog in the order in which they were declared.
Optionally they can be sorted (by specifying ordering constraints) or
randomised. Sorting is usually recommended over relying on declaration
order~\cite{rhett16}.

\input{examples/tex/ex1.tex}

\subsubsection{Variables and strict mode}

Puppet lacks variable declarations in the usual sense; instead variables are
implicitly declared when they are assigned to. A compilation error results if
a given variable is assigned to more than once in the same scope. As we saw
above, unqualified variables, whether being read or assigned to, are written
in ``scripting language'' style \kw{\$$x$}.

Puppet allows variables to be used before they are assigned, in which case
their value is a special ``undefined'' value \kw{undef}, analogous to Ruby's
\kw{nil} or JavaScript's \kw{undefined}. By default, attributes only appear in
the compiled output if their values are defined. Consider the variables
\kw{\$mode} and
\kw{\$checksum} introduced by the assignments at lines~\ref{ex1:var1}
and~\ref{ex1:var2} in the manifest in Figure~\ref{ex1}. The ordering of these
variables relative to the file resource
\kw{config1} is significant, because it affects whether they are in scope.
Since \kw{\$mode} is defined \emph{before}
\kw{config1}, its value can be read and assigned to the attribute
\kw{mode}. In the compiled catalog,
\kw{mode} thus appears as an attribute of \kw{config1}. On the other hand
\kw{\$checksum} is assigned \emph{after} \kw{config1}, and is therefore
undefined when read by the code which initialises the \kw{checksum} attribute.
Thus \kw{checksum} is omitted from the compiled version of \kw{config1}.

Since relying on the values of undefined variables is often considered poor
practice, Puppet provides a \emph{strict} mode which treats the use of
undefined variables as an error. For similar reasons, and also to keep the
formal model simple, \muPuppet always operates in strict mode. We discuss the
possibility of relaxing this in Section~\ref{sec:implementation:discussion}.

\begin{figure}[t]
\centering
  \includegraphics[scale=0.5]{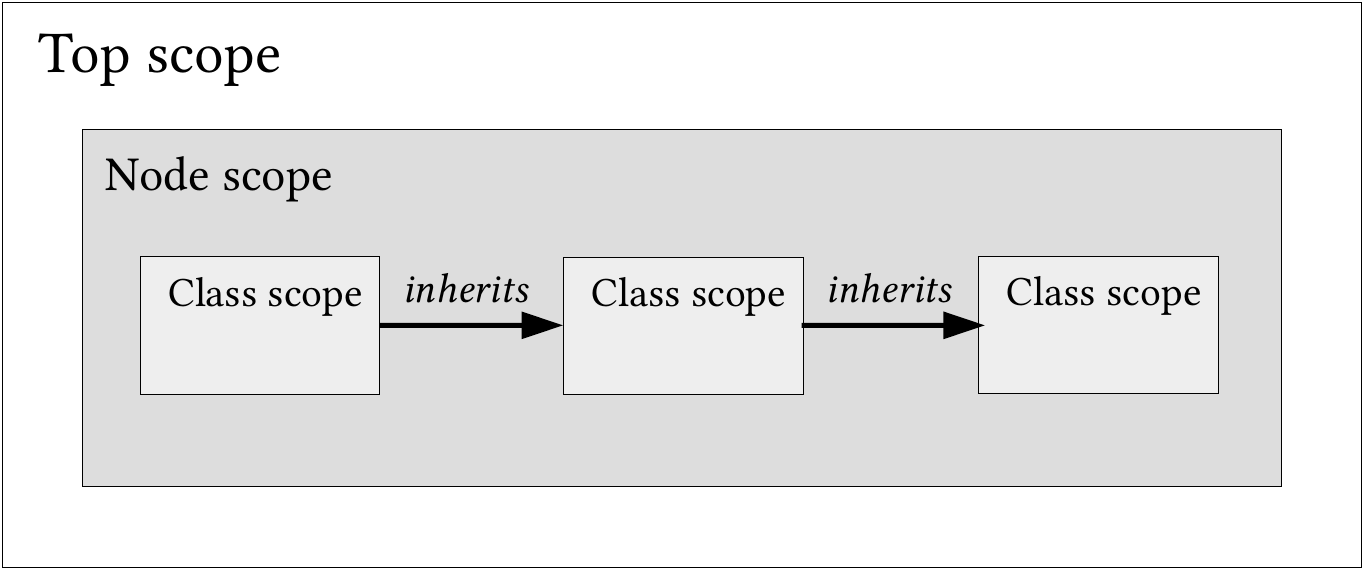}
  \caption{Two aspects of scope: parent scopes (shown as containment), and inheritance chains}
  \label{fig:scopes}
\end{figure}

\subsubsection{Classes and includes}
\label{sec:overview:classes}

Resource declarations may be grouped into {\em classes}. However, Puppet
``classes'' are quite different from the usual concept of classes in
object-oriented programming -- they define collections of resources which can
be declared together by {\em including} the class. This is sometimes called
\emph{declaring} the class, although there is a subtle but important
distinction between ``declaring'' and ``including'' which we will return to
shortly.

In Figure~\ref{ex1}, it is the inclusion into the node definition of class
\kw{service1} which explains the appearance of \kw{config1} in the catalog, and 
in turn the inclusion into \kw{service1} of class \kw{service2} which explains
the appearance of \kw{config2}. (The fact that
\kw{config3} also appears in the output relates to inheritance, and is
discussed in
\S\ref{sec:overview:inheritance} below.) Inclusion is idempotent: the same
class may be included multiple times, but doing so only generates a single
copy of the resources in the catalog. This allows a set of resources to be
included into all locations in the manifest which depend on them, without
causing errors due to duplicate declarations of the same resource.

To a first approximation, including a class into another class obeys a lexical
scope discipline, meaning names in the including class are not visible in the
included class. However inclusion into a node definition has a quite different
behaviour: it introduces a containment relation between the node definition
and the class, meaning that names scoped to the node definition are visible in
the body of the included class. Thus in Figure~\ref{ex1}, although the
variable \kw{\$mode} defined in
\kw{service1} is not in scope inside the included class \kw{service2} (as per
lexical scoping), the \kw{\$source} variable defined in the node definition
\emph{is} in scope in \kw{service1}, because \kw{service1} is included into
the node scope.

This is similar to the situation in Java where a class asserts its membership
of a package using a package declaration, except here the node definition pulls
\emph{in} the classes it requires. The subtlety is that it is actually when a
class is
 \emph{declared} (included for the first time, dynamically speaking) that any
names in the body of the class are resolved. If the \emph{usage} of a class
happens to change so that it ends up being declared in so-called
\emph{top} scope (the root namespace usually determine at check-in time), it
may pick up a different set of bindings. Thus including a class, although
idempotent, has a ``side effect'' -- binding the names in the class -- making
Puppet programs potentially fragile. More of the details of scoping are given
in the language reference manual~\cite{puppet4.8-lang-ref}.
 
\subsubsection{Qualified names}

A definition which is not in scope can be accessed using a \emph{qualified}
name, using a syntax reminiscent of C++ and Java, with atomic names separated
by the token \verb!::!. For example, in Figure~\ref{fig:simple-manifest}
above,
\verb!$::osfamily! refers to a variable in the top scope, while
\verb!$::ssh::params::sshd_package! is an absolute reference to the
\verb!$sshd_package! variable of class
\verb!ssh::params!. 

Less conventionally, Puppet also allows the name of a class to be a qualified
name, such as \verb!ssh::params! in Figure~\ref{fig:simple-manifest}. Despite
the suggestive syntax, which resembles a C++ member declaration, this is
mostly just a convention used to indicate related classes. In particular,
qualified names used in this way do not require any of the qualifying prefixes
to denote an actual namespace. (Although see the discussion in
Section~\ref{sec:implementation:discussion} for an interaction between this
feature and nested classes, which \muPuppet does not support.)

\subsubsection{Inheritance and class parameters}
\label{sec:overview:inheritance}

Classes may {\em inherit} from other classes; the inheriting class inherits
the variables of the parent class, including their values. In the earlier
example (Figure~\ref{ex1}), \kw{service2} inherits the value of
\verb!$provider! from \kw{service3}. Including a derived class implicitly
includes the inherited class, potentially causing the inherited class to be
declared (in the Puppet sense of the word) when the derived class is declared:

\begin{quote}
{\it When you declare a derived class whose base class hasn't already been
  declared, the base class is immediately declared in the current scope, and
  its parent assigned accordingly. This effectively ``inserts'' the base class
  between the derived class and the current scope. (If the base class has
  already been declared elsewhere, its existing parent scope is not changed.)}
\end{quote} 

\noindent This explains why \kw{config3} appears in the compiled catalog for
Figure~\ref{ex1}.

Since the scope in which a class is eventually declared determines the meaning
of the names in the class (\S~\ref{sec:overview:classes} above), inheritance
may have surprising (and non-local) consequences. At any rate, the use of
inheritance for most use cases is now discouraged.\footnote{\url{https://docs
.puppet.com/puppet/latest/style_guide.html}, section 11.1.} The main exception
is the use of inheritance to specify default values; this is the scenario
illustrated in Figure~\ref{fig:simple-manifest}.

\input{fig/simple-manifest}

Line~\ref{ex:ssh::params} of Figure~\ref{fig:simple-manifest} introduces class
\verb!ssh::params!, which assigns to variable
\verb!$sshd_package! a value conditional on the operating system name
\verb!$::osfamily! (line~\ref{ex:osfamily}). The class
\verb!ssh! (line~\ref{ex:default}) inherits from
\verb!ssh::params!. It also defines a \emph{class parameter} \verb!$ssh_pkg!
(before the \kinherits clause), and uses the value of the
\verb!$sshd_package! variable in the inherited class as the default value for
the parameter. Because an inherited class is processed before a derived
class, the value of \verb!$sshd_package! is available at this point.

The value of the parameter \verb!$ssh_pkg! is then used as the title of the
\verb!package! resource declared in the \verb|ssh| class
(line~\ref{ex:package}) specifying that the relevant software package exists
in the target configuration. The last construct is a node definition
specifying how to configure the machine with hostname
\verb!ssh.example.com!. If host \verb!ssh.example.com! is a Debian machine,
the result of compiling this manifest is a catalog containing the following
\verb|package| resource:
\begin{lstlisting}
package { "ssh" :  ensure => installed } 
\end{lstlisting}

\subsubsection{Class statements}

Figure~\ref{ex2} defines a class \kw{c} with three parameters. The \kclass
statement (line~\ref{ex2:cdecl}) can be used to include a class and provide
values for (some of) the parameters. In the resulting catalog, the
\kw{from\_class} resource has \kw{backup} set to \ktrue (from the explicit
argument), \kw{mode} set to \kw{123} (because no \kw{mode} argument is
specified), and \kw{source} set to \kw{'/default'} (because the \kw{path}
variable is undefined at the point where the class is declared
(line~\ref{ex2:cdecl})).

\input{examples/tex/ex2.tex}

However, the potential for conflicting parameter values means that multiple
declarations with parameters are not permitted, and the \kclass statement must
be used instead (which only allows a single declaration).

\subsubsection{Defined resource types}

\emph{Defined resource types} are similarly to classes, but provide a more
flexible way of introducing a user-defined set of resources. Definition \kw{d}
(line~\ref{ex2:ddef}) in Figure~\ref{ex2} introduces a defined resource type.
The definition looks very similar to a class definition, but the body is a
macro which can be instantiated (line~\ref{ex2:dcall}) multiple times with
different parameters.

Interestingly, the \kw{path} attribute in the \kw{from\_class} file is
undefined in the result, apparently because the assignment \kw{\$path =
'/path'} follows the declaration of the class --- however, in the
\kw{from\_define} file, \kw{path} \emph{is} defined as \kw{'/path'}! The
reason appears to be that defined resources are added to the catalog and
re-processed after other manifest
constructs.\footnote{\url{http://puppet-on-the-
edge.blogspot.co.uk/2014/04/getting-your-puppet-ducks-in-row.html}}

%% file: examples/tex/ex1.tex
\begin{figure}[!tb]
\begin{minipage}{.55\linewidth}
\begin{lstlisting}
node default {^\plab{ex1:node}^
  $source = '/source'
  include service1
}

class service1 {
  $mode = 123^\plab{ex1:var1}^

  include service2

  file { 'config1':^\plab{ex1:config1}^
    path => 'path1',
    source => $source,
    mode => $mode,
    checksum => $checksum,
    provider => $provider,
    recurse => $recurse
  }

  $checksum = md5^\plab{ex1:var2}^
}

class service2^\plab{ex1:service2}^ inherits service3 {
  $recurse = true^\plab{ex1:B}^

  file { 'config2':^\plab{ex1:config2}^
    path => 'path2',
    source => $source,
    mode => $mode,
    checksum => $checksum,
    provider => $provider,
    recurse => $recurse
  }
}

class service3 {^\plab{ex1:service3}^
  $provider = posix

  file { 'config3':^\plab{ex1:config3}^
    path => 'path3',
    mode => $mode,
    checksum => $checksum,
    recurse => $recurse
  }
}
\end{lstlisting}
\end{minipage}
\hfill
\begin{minipage}{.40\linewidth}
\begin{lstlisting}
file { 'config3':
    path => 'path3'
}
file { 'config2':
    path => 'path2',
    source => '/source',
    provider => 'posix',
    recurse => true
}
file { 'config1':
    path => 'path1',
    source => '/source',
    mode => 123
}
\end{lstlisting}
\end{minipage}
\caption{Example manifest (left); compiled catalog (right)}
\label{ex1}
\end{figure}

%% file: fig/simple-manifest.tex
\begin{figure}[tb]
\begin{lstlisting}
class ssh::params {^\plab{ex:ssh::params}^
  case $::osfamily {^\plab{ex:osfamily}^
  'Debian': { $sshd_package  = 'ssh' }
  'RedHat': { $sshd_package  = 'openssh-server' }
  default:  { fail("SSH class not supported") }
  }
}
class ssh ($ssh_pkg = $::ssh::params::sshd_package) inherits ssh::params { ^\plab{ex:default}^
  package { $ssh_pkg: ^\plab{ex:package}^
    ensure => installed
  }
}
node 'ssh.example.com' {
  include ssh ^\plab{ex:include}^
}
\end{lstlisting}

  \caption{Example manifest showing recommended use of inheritance for setting default parameters}
  \label{fig:simple-manifest}
\end{figure}

%% file: examples/tex/ex2.tex
\begin{figure}[!htb]
\begin{minipage}{.55\linewidth}
\begin{lstlisting}
class c (^\plab{ex2:service2}^
  $backupArg = false,
  $pathArg = '/default',
  $modeArg = 123 ) {

  file { 'from_class':^\plab{ex2:config2}^
    backup => $backupArg,
    source => $pathArg,
    path => $path,
    mode => $modeArg
  }
}

define d (^\plab{ex2:ddef}^
  $backupArg = false,
  $pathArg = '/default',
  $modeArg = 123 ) {

  file { 'from_define':^\plab{ex2:config3}^
    backup => $backupArg,
    source => $pathArg,
    path => $path,
    mode => $modeArg
  }
}

node default {

  $backup = true

  class { c:^\plab{ex2:cdecl}^
    backupArg => $backup,
    pathArg => $path 
  }

  d { "service3":^\plab{ex2:dcall}^
    backupArg => $backup,
    pathArg => $path
  }

  $path = '/path'
}
\end{lstlisting}
\end{minipage}
\hfill
\begin{minipage}{.40\linewidth}
\begin{lstlisting}
file { 'from_class':
    backup => true,
    source => '/default',
    mode => 123
}
file { 'from_define':
    path => '/path',
    backup => true,
    source => '/default',
    mode => 123
}
\end{lstlisting}
\end{minipage}
\caption{Manifest with class parameters and defined resource types (left); catalog (right)}
\label{ex2}
\end{figure}

%% file: sec/semantics.tex
\section{\muPuppet}
\label{sec:language}

We now formalise \muPuppet, a language which captures many of the essential
features of Puppet.  Our goal is not to model all of Puppet's
idiosyncrasies, but instead to attempt to capture the `declarative'
core of Puppet, as a starting point for future study.  As we discuss
later, Puppet also contains several non-declarative features whose
behaviour can be counterintuitive and surprising; their
use tends to be discouraged in Puppet's documentation and by other
authors~\cite{rhett16}.

\subsection{Abstract syntax}
\label{sec:syntax}

The syntax of \muPuppet manifests $m$ is defined in
Figure~\ref{fig:muppet}, including expressions $e$ and statements
$s$. Constant expressions in \muPuppet can be integer literals $i$,
string literals $w$, or boolean literals $\ktrue$ or $\kfalse$. Other
expressions include arithmetic and boolean operations, variable forms
$\$x$, $\$\scopeSep x$ and $\$\classScope{a}\scopeSep x$.  Here, $x$
stands for variable names and $a$ stands for class names.  
\emph{Selectors} $e\ ?\ \{M\}$ are conditional expressions that evaluate $e$ and then
conditionally evaluate the first matching branch in $M$.  Arrays
are written $[e_1,\ldots,e_n]$ and hashes (dictionaries) are written $\{k \Rightarrow
e,\ldots\}$ where $k$ is a key (either a constant number or string). A
reference  $e_1[e_2] $ describes an array, a hash or a resource reference where $e_1$ itself can be a reference. 
When it is a resource reference, $e_{1}$ could be a built-in resource type.
Full Puppet includes additional base types
(such as floating-point numbers) and many more built-in functions that
we omit here.

\input{fig/syntax}

Statements $s$ include expressions $e$ (whose value is discarded),
composite statements $s_1 \seq s_2$, assignments $\$x=e$, and
conditionals \kunless, \kif, \kcase, which are mostly standard. (Full
Puppet includes an \kelseif construct that we omit from \muPuppet.)
Statements also include resource declarations $t\;\{e:H\}$ for
built-in resource types $t$, resource declarations $u\ \{e:H\}$ for
defined resource types $u$, and class declarations
$\kclass\ \{a: H\}$ and $\kinclude~a$.  

Manifests $m$ can be statements $s$; composite manifests
$m_1 \seq m_2$, class definitions $\kclass\ a \ \{s\}$ with or without
parameters $\rho$ and inheritance clauses $\kinherits\ b$; node
definitions $\knode~Q~\{s\}$; or defined resource type definitions
$\kdefine~u~(\rho)~\{s\}$. Node specifications $Q$ include literal
node names $N$, \kdefault, lists of node names, and regular
expressions $r$ (which we do not model explicitly).

Sequences of statements, cases, or manifest items can be written by writing
one statement after the other, separated by whitespace, and we write
$\visiblespace$ when necessary to emphasise that this whitespace is
significant.  The symbol $\varepsilon$ denotes the empty string.

\subsection{Operational Semantics}
\label{sec:semantics}

We now define a small-step operational semantics for \muPuppet.  This
is a considered choice: although Puppet is advertised as a
declarative language, it is not \emph{a priori} clear that manifest
compilation is a terminating or even deterministic process.  Using
small-step semantics allows us to translate the (often) procedural
descriptions of Puppet's constructs directly from the documentation.

The operational semantics relies on auxiliary notions of catalogs
$\CV$, scopes $\alpha$, variable environments $\sigma$, and definition
environments $\kappa$ explained in more detail below.  We employ three
main judgements, for processing expressions, statements, and
manifests:
\[\sigma, \kappa, \CV,e \stepExpr{\alpha} e' 
  \qquad 
  \sigma,\kappa,\CV, s \stepStmt{\alpha} \sigma',\kappa',\CV' , s'
  \qquad 
  \sigma,\kappa,\CV, m \stepManifest{N} \sigma',\kappa',\CV' , m'\]
Here, $\sigma$, $\kappa$, and $\CV$ are the variable environment, definition
environment, and catalog beforehand, and their primed versions are the
corresponding components after one compilation step. The parameter $\alpha$
for expressions and statements represents the ambient scope; the parameter $N$
for manifests is the target node name.

The main judgement is $\stepManifest{}$, which takes a \muPuppet manifest $m$
and a node name $N$ and compiles it to a catalog $\CV$, that is, a list of
resource values $\RV$ for that node. Given initial variable environments
$\sigma$ (representing data provided by the client) and $\kappa$ (containing
any predefined classes or resource definitions), execution of manifest $m$
begins with an empty catalog and terminates with catalog $\CV$ when the
manifest equals $\kskip$, i.e. $\sigma,\kappa,\varepsilon , m \stepManifest{N}
\cdots
\stepManifest{N} \sigma',\kappa',\CV, \kskip$.

\subsubsection{Auxiliary definitions: catalogs, scopes and environments}

Before defining compilation formally, we first define catalogs
(\S\ref{sec:catalog}), the result of compiling manifests; scopes
(\S\ref{sec:scope}), which explicitly represent the ambient scope used to
resolve unqualified variable references; variable environments
(\S\ref{sec:var-environment}), which store variable bindings; and definition
environments (\S\ref{sec:def-environment}), which store class and resource
definitions.

\paragraph{Catalogs}
\label{sec:catalog}

The syntax of catalogs is given in Figure~\ref{fig:runtime-syntax}. A
\emph{catalog} $\CV$ is a sequence of resource values, separated by
whitespace; a \emph{resource value} $\RV = t\;\{w : \HV\}$ is a
resource whose
title is a string value and whose content is a hash value; a hash value $\HV$ is
an attribute-value sequence in which all expressions are values; and finally a
\emph{value} $v$ is either an integer literal $i$, string literal $w$, 
boolean literal $\ktrue$ or $\kfalse$, hash $\{\HV\}$, array
$[v_1,\ldots,v_n]$ or a reference value $t[v]$.  In a well-formed catalog, there is at most one
resource with a given type and title; attempting to add a resource
with the same type and title as one already in the catalog is an error.

\input{fig/runtime-syntax}

\paragraph{Scopes}
\label{sec:scope}

As discussed in Section~\ref{sec:overview}, Puppet variables can be
assigned in one scope and referenced in a different scope. For example,
in Figure~\ref{fig:simple-manifest}, the parent scope of class scope
\verb|ssh| is class scope \verb|ssh::params|. To model this, we model
scopes and parent-child relations between scopes explicitly. Scope $\topScope$
represents the top scope, $\classScope{a}$ is the scope of class $a$,
$\nodeScope$ is the active node scope, and $\defScope{\alpha}$ is the scope of
a resource definition that is executed in ambient scope $\alpha$.

The scope for defined resources takes another scope parameter $\alpha$ in
order to model resource definitions that call other resource definitions. The
top-level, class, and node scopes are persistent, while $\defScope{\alpha}$ is
cleared at the end of the corresponding resource definition; thus these scopes
can be thought of as names for stack frames. The special statement form
$\kscope~\alpha~s$ is used internally in the semantics to model scope changes.
An additional internal statement form \kskip, unrelated to scopes, represents
the empty statement. Neither of these forms are allowed in Puppet manifests.

As discussed earlier, there is an ancestry relation on scopes, which
governs the order in which scopes are checked when dereferencing an
unqualified variable reference. We use mutually recursive auxiliary judgments $\alpha
\parentof{\kappa} \beta$
to indicate that $\alpha$ is the parent scope of $\beta$ in the context of
$\kappa$ and $\alpha \baseof{\kappa} \beta$ to indicate that $\alpha$ is the
\emph{base} scope of $\beta$.  The base scope is either $\topScope$,
indicating that the scope is the top scope, or $\nodeScope$, indicating that
the scope is being processed inside a node definition.  We first discuss the
rules for $\parentof{\kappa}$:
\input{rule/environment/parentof}
The \textsc{PNode} rules says that the top-level scope is the parent
scope of node scope.  The \textsc{PDefRes} rule says that the parent
scope of the defined resource type scope is its base scope.  Thus, a
resource definition being declared in the toplevel will have parent
$\topScope$, while one being declared inside a node definition will
have parent scope $\nodeScope$. The \textsc{PClass} rule defines the scope of the
(declared) parent class to be the scope $\alpha$ that is recorded in the
$\Declared(\alpha)$ entry.  The rules for class inclusion and
declaration in the next section show how the $\Declared(\alpha)$ entry
is initialised; this also uses the $\baseof{\kappa}$ relation.
The rules defining $\baseof{\kappa}$ are as follows:
\input{rule/environment/baseof}
These rules determine whether the ambient scope $\alpha$ in which the class is
declared is \emph{inside} or \emph{outside} a node declaration.  The base
scope of toplevel or node scope is toplevel or node scope respectively.  The
base scope of $\defScope{\beta}$  is the base scope of $\beta$, while the
base scope of a class scope $\classScope{a}$ is the base scope of its parent
scope as defined in the definition environment $\kappa$. (We will only try to
obtain the base scope of a class that has already been declared.)

\paragraph{Variable environments}
\label{sec:var-environment}

During the compilation of a manifest, the values of variables are recorded in
\emph{variable environments} $\sigma$ which are then accessed when these
variables are referenced in the manifest. (We call these \emph{variable}
environments, rather than plain environments, since ``environment'' has a
specific technical meaning in Puppet; see the glossary in \appref{glossary}.) As 
discussed in section 2.1.3, Puppet allows variables to be referenced before being defined, whereas
the variable environment is designed in the way not to allow it. A variable can only be referenced in 
the environment if its value has been stored. Thus the undefined variables in the manifest in Figure 2 
are not legal in \muPuppet.

Formally, a variable environment is defined as a partial function $\sigma :
\Set{Scope}
\times \Set{Var} \to \Set{Value}$ which maps pairs of scopes and
variables to values. The scope component indicates the scope in which the
variable was assigned. We sometimes write $\sigma_\alpha(x)$ for
$\sigma(\alpha,x)$. Updating a variable environment $\sigma$ with new binding
$(\alpha,x) $ to $v$ is written $\fupdate{\sigma}{\alpha}{x}{v}$, and clearing
an environment (removing all bindings in scope $\alpha$) is written
$\clear(\sigma,\alpha)$.

\paragraph{Definition environments}
\label{sec:def-environment}

Some components in Puppet, like classes and defined resource types,
introduce \emph{definitions} which can be declared elsewhere. To model
this, we record such definitions in \emph{definition environments}
$\kappa$. Formally, a definition environment is a partial function
$\kappa : \Set{Identifier} \to \Set{Definition}$ mapping each
identifier to a definition $D$.  Evaluation of the definition only
begins when a resource is declared which uses that definition.

Definitions are of the following forms:
\begin{eqnarray*}
D &\bnf&  \ClassDef(c_{opt}, \rho, s)
\mid  \Declared(\alpha)
\mid \ResourceDef(\rho, s) \\
c_{opt} &\bnf& c \mid \bot
\end{eqnarray*}
The definition form $\ClassDef(c_{opt}, \rho, s)$ represents the
definition of a class (before it has been declared); $c_{opt}$ is the
optional name of the class's parent, $\rho$ is the list of parameters
of the class (with optional default values), and $s$ is the body of
the class. The definition form $\Declared(\alpha)$ represents a class
that has been declared; $\alpha$ is the class's \emph{parent scope}
and $\rho$ and $s$ are no longer needed. In Puppet, the definition of a class
can appear before or after its declaration, as we saw in
the manifest in Figure 2, whereas the definition environment is
designed to require that a class is defined before it is declared. Thus the
inclusion of class \kw{service1} in Figure 2 will be not evaluated in \muPuppet.
Moreover, a class can be declared only once in Puppet, and when it is declared 
its definition environment entry is changed to $\Declared(c_{opt})$. Finally, the definition form
$\ResourceDef(\rho, s)$ represents the definition of a new
\emph{resource type}, where $\rho$ and $s$ are as above.

\subsubsection{Expression evaluation}

Expressions are the basic computational components of \muPuppet. The rules for
expression forms such as primitive operations are standard. The rules for
selector expressions are also straightforward. Since variable accessibility
depends on scope, the variable evaluation rules are a little more involved:
\input{rule/expression/variable}

The \textsc{LVar} looks up the value of an unqualified variable in the current
scope, if present. The \textsc{PVar} rule handles the case of an unqualified
variable that is not defined in the current scope; its value is the value of
the variable in the parent scope. The \textsc{TVar} and \textsc{QVar} rules
look up fully-qualified variables in top scope or class scope, respectively.
(There is no qualified syntax for referencing variables in node scope from
other scopes.)

$\mu$Puppet also includes array and hash expressions.  An array is a
list of expressions in brackets and a hash is a list of keys and their
expression assignments in braces. When the expressions are values, an
array or a hash is also a value. Each expression in them can be
dereferenced by the array or hash followed by its index or key in
brackets. The rules for constructing and evaluating arrays and hashes
are straightforward, and included in \appref{rules}.

Resource references of the form $t[v]$ are allowed as values, where
$t$ is a built-in resource name and $v$ is a (string) value.  Such
references can be used as parameters in other resources and to express
ordering relationships between resources.  Resource references can be
used to extend resources or override inherited resource parameters; we
do not model this behaviour.  A resource reference can also (as of
Puppet 4) be used to access the values of the resource's parameters.
This is supported in $\mu$Puppet as shown in the following
example. 

\begin{lstlisting}
  file {"foo.txt": 
    owner => "alice"
  }
  $y = "foo.txt"
  $x = File[$y]
  file {"bar.txt":
    owner => $x["owner"]
  }
\end{lstlisting}
In this example, we first declare a file resource, with an
\verb|owner| parameter \verb|"alice"|, then we assign $y$ the filename
and $\$x$ a resource reference (value) \verb|File["foo.txt"]|.  Then
in defining a second file resource we refer to the \verb|"owner"|
parameter of the already-declared file resource via the reference
\verb|File["foo.txt"]|. This declaration results in a second file
resource with the same owner as the first.

The rules for dereferencing arrays, hashes, and resource references
are as follows:
\input{rule/expression/dereference}

In the rule \textsc{DeRefExp} the expression $e$ is evaluated to an array or a hash value. The rule \textsc{DeRefIndex} evaluates the index inside the brackets to a value. Rule \textsc{DeRefArray} accesses the value in an array at the index $n$ while rule \textsc{DeRefHash} accesses the hash value by searching its key $k$.  There could be a sequence of reference indexes in a reference. As we can see, such reference is evaluated in the left-to-right order of the index list. 
Rule \textsc{ResRef} evaluates the index and in the
\textsc{DeRefRec} rule, the function $\lookupCat$ looks up the
catalog for the value of the attribute $k$ of the resource $t[v]$.

\subsubsection{Statement evaluation}
As with expressions, some of the statement forms, such as sequential
composition, conditionals (\verb|if|, \verb|unless|), and \verb|case|
statements have a conventional operational semantics, shown in \appref{rules}.
An expression can occur as a statement; its value is ignored. Assignments,
like variable references, are a little more complex. When storing the value of
a variable in an assignment in $\sigma$, the compilation rule binds the value
to $x$ in the scope $\alpha$:
\input{rule/statement/assignment}

\noindent Notice that Puppet does not allow assignment into any other scopes,
only the current scope $\alpha$.

We now consider $\kscope~\alpha~s$ statements, which are internal
constructs (not part of the Puppet source language) we have introduced
to track the scope that is in effect in different parts of the
manifest during execution. The following rules handle compilation
inside \kscope statements and cleanup when execution inside such a
statement finally terminates.

\input{rule/statement/scope}

The \textsc{ScopeStep} and \textsc{ScopeDef} rules handle compilation
inside a scope; the ambient scope $\alpha'$ is overridden and the
scope parameter $\alpha$ is used instead.  The \textsc{ScopeDone} rule
handles the end of compilation inside a ``persistent'' scope, such as
top-level, node or class scope, whose variables persist throughout
execution, and the \textsc{DefScopeDone} rule handles the temporary
scope of defined resources, whose locally-defined variables and
parameters become unbound at the end of the definition.  (In contrast,
variables defined in toplevel, node, or class scopes remain visible
throughout compilation.)

Resource declarations are compiled in a straightforward way; the title
expression is evaluated, then all the expressions in attribute-value pairs in
the hash component are evaluated. Once a resource is fully evaluated, it is
appended to the catalog:
\input{rule/statement/resource-declaration}

Defined resource declarations look much like built-in resources:
\input{rule/statement/defined-resource-type}
The $merge$ function returns a statement $s'$ assigning the parameters to
their default values in $\rho$ or overridden values from $\HV$. Notice
that we also add the special parameter binding $\$title = w$; this
is because in Puppet, the title of a defined resource is made
available in the body of the resource using the parameter $\$title$.  The body of
the resource definition $s$ is processed in scope $\defScope{\alpha}$. Class
declarations take two forms: \emph{include-like} and
\emph{resource-like declarations}. 

The statement $\kinclude\ a$ is an include-like
declaration of a class $a$.  (Puppet includes some additional
include-like declaration forms such as \verb|contain| and
\verb|require|).  Intuitively, this means that the class body is
processed (declaring any ancestors and resources inside the class),
and the class is marked as declared; a class can be declared at most
once. The simplest case is when a class has no parent, covered by the
first two rules below:
\input{rule/statement/include}
In the \textsc{IncU} rule, the class has not been declared yet, so we look up its body and
default parameters and process the body in the appropriate scope. (We use the
$merge$ function again here to obtain a statement initialising all parameters
which have default values.)  In addition, we modify the class's entry
in $\kappa$ to $\Declared(\beta)$, where $\beta \baseof{\kappa}
\alpha$.  
As described in Section~\ref{sec:overview}, this aspect of
Puppet scoping is dynamic: if a base class is defined outside a node
definition then its parent scope is $\topScope$, whereas if it is
declared during the processing of a node definition then its parent
scope is $\nodeScope$.  (As discussed below, if a class inherits from another,
however, the parent scope is the scope of the parent class no matter
what).  If this sounds confusing, this is because it is; this is the
trickiest aspect of Puppet scope that is correctly handled by
\muPuppet.  This complexity appears to be one reason that the use of
node-scoped variables is discouraged
by some experts~\cite{rhett16}.

In the \textsc{IncD} rule, the class $a$ is already declared, so no action
needs to be taken.  In the \textsc{IncPU} rule, we include the parent class so
that it (and any ancestors) will be processed first.  If there is an
inheritance cycle, this process loops; we have confirmed experimentally that
Puppet does not check for such cycles and instead fails with a stack overflow.
In the \textsc{IncPD} rule, the parent class is already declared, so we
proceed just as in the case where there is no parent class.

The rules for \emph{resource-like class declarations} are similar:
\input{rule/statement/resource-like-class-declaration}
There are two differences.  First, because resource-like class
declarations provide parameters, the rule \textsc{CDecStep} provides
for evaluation of these parameters.  Second, there is no rule
analogous to \textsc{IncD} that ignores re-declaration of an
already-declared class.  Instead, this is an error.  (As with multiple
definitions of variables and other constructs, however, we do not
explicitly model errors in our rules.)

\subsubsection{Manifest compilation}

At the top level, manifests can contain statements, node definitions, resource
type definitions, and class definitions. To compile statements at the top
level, we use the following rule:
\input{rule/manifest/statement}
The main point of interest here is that we change from the manifest
judgement (with the node name parameter $N$) to the statement
judgement (with toplevel scope parameter $\topScope$).  The node name parameter
is not needed for processing statements, and we (initially) process
statements in the toplevel scope.  Of course, the statement $s$ may
well itself be a \kscope statement which immediately changes the
scope.

A manifest in Puppet can configure all the machines (nodes) in a
system. A node definition describes the configuration of one node (or
type of nodes) in the system. The node declaration includes a
specifier $Q$ used to match against the node's hostname. We abstract
this matching process as a function $\nodeMatch(N,Q)$ that checks if
the name $N$ of the requesting computer matches the specifier
$Q$. If so (\textsc{NodeMatch}) we will compile the statement body of
$N$. Otherwise (\textsc{NodeNoMatch}) we will skip this definition and
process the rest of the manifest.
\input{rule/manifest/node-definition}

Resource type definitions in Puppet are used to design new, high-level
resource types, possibly by declaring other built-in resource types,
defined resource types, or classes. Such a definition includes Puppet
code to be executed when the a resource of the defined type is
declared.  Defined resource types can be declared multiple times with
different parameters, so resource type definitions are loosely
analogous to procedure calls.  
The following is an example of a defined resource type:
\begin{lstlisting}
define apache::vhost (Integer $port) {
  include apache
  file { "host":
    content => template('apache/vhost-default.conf.erb'),
    owner   => 'www'
  }
}
\end{lstlisting}
The compilation rule for defining a defined resource type is:
\input{rule/manifest/defined-resource-type}
The definition environment is updated to map $u$ to $\ResourceDef(\rho,s)$
containing the parameters and statements in the definition of $u$. The
manifest then becomes $\kskip$.

A class definition is used for specifying a particular service that
could include a set of resources and other statements. Classes are
defined at the top level and are declared as part of statements, as
described earlier.  Classes can be parameterised; the parameters are
passed in at declaration time using the resource-like declaration
syntax. The parameters can be referenced as variables in the class
body. A class can also inherit directly from one other class. The
following rules handle the four possible cases:
\input{rule/manifest/class-definition}
In the simplest case (\textsc{CDef}) we add the class definition to the
definition environment with no parent and no parameters.  The other three
rules handle the cases with inheritance, with parameters, or with both.

%% file: fig/syntax.tex
\begin{figure}[tb]
\small
\[\begin{array}{lrcl}
\text{Expression}& e&::=&i \mid w \mid \ktrue \mid \kfalse \mid
                          \$x\mid\$\scopeSep
                          x\mid\$\classScope{a}\scopeSep x\\
&& \mid&  e_{1}+e_{2}\mid e_{1}-e_{2}\mid e_{1}/e_{2}
\mid e_{1}>e_{2}\mid e_{1}=e_{2} \mid e_{1}\kand e_{2} \mid e_{1}\kor
         e_{2} \mid \Not{e} \mid\ldots\\
&&\mid&  \{H\} \mid [e_1,\ldots,e_n] \mid e_1[e_2] \mid e\ \query\ \{M\} \\
\text{Array}&A&::=&\varepsilon \mid e, A\\
\text{Hash}&H&::=& \varepsilon \mid k\Rightarrow e, H\\
\text{Case}&c&::=& e \mid \kdefault \\
\text{Matches}& M&::=&\varepsilon \mid c\Rightarrow e, M\medskip\\
\text{Statement}& s&::=&e\mid s_1 \seq s_2 \mid \$x=e \mid
                          \kunless\ e\ \{s\} \mid \kif\  e\
                          \{s\}\ \kelse~\{s\} \mid \kcase\ e\ \{C\} \mid D \\
\text{Cases}&C&::=& \varepsilon \mid c: \{s\}\visiblespace C \\
\text{Declaration}&D&::=& t\;\{e:H\}\mid u\ \{e:H\} \mid\kclass\ \{a:H\} \mid \kinclude\ a
\medskip\\
\text{Manifest}&m&::=&s \mid m_1 \seq m_2 \mid \knode\ Q\ \{s\} \mid
                       \kdefine\ u\ (\rho)\ \{s\}\mid
                       \kclass\ a \ \{s\}\mid
        \kclass\ a\ (\rho)\ \{s\}
\\
&&\mid& \kclass\ a \ \kinherits\ b \ \{s\}\mid
        \kclass\ a\ (\rho)\ \kinherits\ b\ \{s\}\\
\text{Node spec} & Q & \bnf & N \mid \kdefault \mid (N_1,\ldots,N_k) \mid r \in \Set{RegExp} \\
\text{Parameters}&\rho & \bnf& \varepsilon \mid \$x, \rho \mid \$x=e, \rho\\
\end{array}
\]
\caption{Abstract syntax of \muPuppet}
\label{fig:muppet}
\end{figure}

%% file: fig/runtime-syntax.tex
\begin{figure}[tb]
\small
\[\begin{array}{lrcl}
\text{Catalog} & \CV & \bnf &\varepsilon \mid \RV\visiblespace\CV
\medskip\\
\text{Value}& v&::=&i\mid w \mid \ktrue \mid \kfalse \mid \{\HV\} \mid
    [v_1,\ldots,v_n]\mid t[v]\\
\text{Hash value} &\HV&::=& \varepsilon \mid k \Rightarrow v,\HV\\
\text{Resource value}& \RV&::=&t\;\{w : \HV\}
\medskip\\
\text{Scope}&\alpha &::=& \topScope \mid \classScope{a} \mid \nodeScope \mid \defScope{\alpha}
\medskip\\
\text{Statement}& s&::=& ... \mid \kscope\;\alpha\;s \mid \kskip
\end{array}
\]
\caption{Auxiliary constructs: catalogs and scopes}
\label{fig:runtime-syntax}
\end{figure}

%% file: rule/environment/parentof.tex
\begin{mathpar}
  \inferrule*[right=PNode]
  {
  }
  {
    \topScope \parentof{\kappa} \nodeScope
  }
  \and
  \inferrule*[right=PDefRes]
  {
    \beta \baseof{\kappa} \defScope{\alpha} 
  }
  {
    \topScope \parentof{\kappa} \defScope{ \alpha}
  }
  \and
  \inferrule*[right=PClass]
  {
    \kappa(a)=\Declared(\alpha)
  }
  {
    \alpha \parentof{\kappa} \classScope{a}
  }
\end{mathpar}

%% file: rule/environment/baseof.tex
\begin{mathpar}
  \inferrule*[right=BTop]
  {
  }
  {
    \topScope \baseof{\kappa} \topScope
  }
  \and
  \inferrule*[right=BNode]
  {
  }
  {
    \nodeScope \baseof{\kappa} \nodeScope
  }
  \and
  \inferrule*[right=BDefRes]
  {
    \alpha \baseof{\kappa} \beta
  }
  {
    \alpha \baseof{\kappa} \defScope{ \beta}
  }
  \and
  \inferrule*[right=BClass]
  {
    \kappa(a)=\Declared(\beta) \\
    \alpha \baseof{\kappa} \beta
  }
  {
    \alpha \baseof{\kappa} \classScope{a}
  }
\end{mathpar}

%% file: rule/expression/variable.tex
\begin{mathpar}
  \inferrule*[right=LVar]
  {
    x\in \dom(\sigma_{\alpha})
  }
  {
    \sigma,\kappa,\CV,\$x\stepExpr{\alpha}\sigma_{\alpha}(x)
  }
  \and
  \inferrule*[right=PVar]
  {
    x\notin \dom(\sigma_{\alpha})
    \\
    \sigma,\kappa,\CV,\$x\stepExpr{\beta} v
    \\
    \beta \parentof{\kappa} \alpha
  }
  {
    \sigma,\kappa,\CV,\$x\stepExpr{\alpha}v
  }
  \and
  \inferrule*[right=TVar]
  {
    x \in \dom(\sigma_{\topScope})
  }
  {
    \sigma,\kappa,\CV, \${::} x\stepExpr{\alpha}\sigma_{\topScope}(x)
  }
  \and
  \inferrule*[right=QVar]
  {
    x\in \dom(\sigma_{\classScope{a}})
  }
  {
    \sigma,\kappa,\CV,\$\classScope{a}::x\stepExpr{\alpha}\sigma_{\classScope{a}}(x)
  }
\end{mathpar}

%% file: rule/expression/dereference.tex
\begin{mathpar}
 \inferrule*[right=DeRefExp]
 {
   \sigma, \kappa, \CV, d \xrightarrow{\alpha} d'
  }
  {
    \sigma,\kappa, \CV, d[e] \xrightarrow{\alpha} d'[e]
  }
  \and
  \inferrule*[right=DeRefIndex]
 {
    \sigma, \kappa,\CV, e \xrightarrow{\alpha} e'
  }
  {
    \sigma,\kappa, \CV, v[e] \xrightarrow{\alpha} v[e']
  }
  \and
  \inferrule*[right=DeRefArray]
 { }
  {
   \sigma, \kappa, \CV, [v_{0},\ldots,v_{n}, \ldots, v_{m}][n] \xrightarrow{\alpha} v_{n}
  }
  \and
  \inferrule*[right=DeRefHash]
 {
    k=k_{n}
  }
  {
   \sigma,\kappa, \CV, \{k_{1}=v_{1}, \ldots,k_{n}=v_{n},\ldots, k_{m}=v_{m}\}[k] \xrightarrow{\alpha}  v_{n}
  } 
  \\
  \and
  \inferrule*[right=RefRes]
 {
   \sigma, \kappa, \CV, e \xrightarrow{\alpha} e'
  }
  {
   \sigma,\kappa, \CV, t[e]\xrightarrow{\alpha} t[e']
  }   
  \and
  \inferrule*[right=DeRefRes]
 {
  \lookupCat(\CV, t, w, k)=v
  }
  {
   \sigma,\kappa, \CV, t[w][k]\xrightarrow{\alpha} v
  }   
\end{mathpar}

%% file: rule/statement/assignment.tex
\begin{mathpar}
  \inferrule*[right=AssignStep]
  {
    \sigma, \kappa , \CV, e \stepExpr{\alpha} e'
  }
  {
    \sigma,\kappa,\CV , \$x = e \stepStmt{\alpha} \sigma,\kappa,\CV, \$x = e'
  }
  \and
  \inferrule*[right=Assign]
  {
    x\notin \dom(\sigma_\alpha)
  }
  {
    \sigma,\kappa,\CV , \$x = v \stepStmt{\alpha}
  \fupdate{\sigma}{\alpha}{x}{v},\kappa,\CV , \kskip
  }
\end{mathpar}

%% file: rule/statement/scope.tex
\begin{mathpar}
  \inferrule*[right=ScopeStep]
  {
    \alpha\in\{\topScope,\classScope{a},\nodeScope\}
    \\
    \sigma,\kappa,\CV, s\stackrel{\alpha}{\rightarrow}_{\mathrm{s}}\sigma',\kappa',\CV', s'
  }
  {
    \sigma,\kappa,\CV, \mathtt{scope}\ \alpha\ s\stepStmt{\alpha'}\sigma',\kappa',\CV', \mathtt{scope}\ \alpha\ s'
  }
  \and
  \inferrule*[right=DefScopeStep]
  {
    \sigma,\kappa,\CV, s\stepStmt{\defScope{\alpha}}\sigma',\kappa',\CV', s'
  }
  {
    \sigma,\kappa,\CV, \mathtt{scope}\ (\defScope{\alpha})\ s\stepStmt{\alpha}\sigma',\kappa',\CV', \mathtt{scope}\ (\defScope{\alpha})\ s'
  }
  \and
  \inferrule*[right=ScopeDone]
  {
    \alpha\in\{\topScope,\; \classScope{a},\;\nodeScope\}
  }
  {
    \sigma,\kappa,\CV , \mathtt{scope}\ \alpha\ \mathtt{skip} \stepStmt{\beta} \sigma,\kappa,\CV
    , \mathtt{skip}
  }
  \and
  \inferrule*[right=DefScopeDone]
  {
  }
  {
    \sigma,\kappa,\CV, \mathtt{scope}\ (\defScope{\alpha})\
  \mathtt{skip}\stepStmt{\alpha}\clear(\sigma,\defScope{\alpha}),\kappa,\CV,\mathtt{skip}
  }
\end{mathpar}

%% file: rule/statement/resource-declaration.tex
\begin{mathpar}
  \inferrule*[right=ResStep]
  {
    \sigma, \kappa , \CV, e:H \stepRes{\alpha}e':H'
  }
  {
    \sigma,\kappa,\CV, t\;\{e:H\} 
    \stepStmt{\alpha} 
    \sigma,\kappa,\CV, t\;\{e':H'\}
  }
  \and
  \inferrule*[right=ResDecl]
  {
    \strut
  }
  {
    \sigma,\kappa,\CV , \RV \stepStmt{\alpha} \sigma,\kappa,\CV \seq \RV
  , \kskip
  }
\end{mathpar}

%% file: rule/statement/defined-resource-type.tex
\begin{mathpar}
  \inferrule*[right=DefStep]
  {
    \sigma,\kappa, \CV, \{e: H\}\stepRes{\alpha} \{e': H'\}
  }
  {
    \sigma,\kappa,\CV, u\ \{e: H\}\stepStmt{\alpha} \sigma,\kappa,\CV, u\ \{e': H'\}
  }
  \and
  \inferrule*[right=Def]
  {
    \kappa(u)=\ResourceDef(\rho, s)
    \\
    s'=\merge(\rho,\HV)
  }
  {
    \sigma,\kappa,\CV, u\ \{w: \HV\}
    \stepStmt{\alpha}
    \sigma,\kappa,\CV, \mathtt{scope}\ (\defScope{\alpha})\ \{\$title= w\seq s' \seq s\}
  }
\end{mathpar}

%% file: rule/statement/include.tex
\begin{mathpar}
  \inferrule*[right=IncU]
  {
    \kappa(a)=\ClassDef(\bot,\rho,s)
    \\
    s' = \merge(\rho,\varepsilon)
    \\
    \beta \baseof{\kappa} \alpha
  }
  {
    \sigma,\kappa,\CV, \mathtt{include}\ a \stepStmt{\alpha}
 \sigma,\defupdate{\kappa}{a}{\Declared(\beta)},\CV,
         \mathtt{scope}\ (\classScope{a})\ \{s' \seq s\}
  }
  \and
  \inferrule*[right=IncD]
  {
    \kappa(a)=\Declared(\beta)
  }
  {
    \sigma,\kappa,\CV,\mathtt{include}\ a
  \stepStmt{\alpha}\sigma,\kappa,\CV,\mathtt{skip}
  }
  \and
  \inferrule*[right=IncPU]
  {
    \kappa(a)=\ClassDef(b,\rho,s)
    \\
    \kappa(b)=\ClassDef(c_{opt},\rho',s')
  }
  {
    \sigma,\kappa,\CV, \mathtt{include}\ a
  \stepStmt{\alpha}\sigma,\kappa,\CV,\mathtt{include}\ b\visiblespace
  \mathtt{include}\ a
  }
  \and
  \inferrule*[right=IncPD]
  {
    \kappa(a)=\ClassDef(b,\rho,s)
    \\
    \kappa(b)=\Declared(\beta)
    \\
    s' = \merge(\rho,\varepsilon)
  }
  {
    \sigma,\kappa,\CV, \mathtt{include}\ a
  \stepStmt{\alpha}\sigma,\defupdate{\kappa}{a}{\Declared(\classScope{b})},\CV,
  \mathtt{scope}\ (\classScope{a})\ \{s' \seq s\}
  }
\end{mathpar}

%% file: rule/statement/resource-like-class-declaration.tex
\begin{mathpar}
  \inferrule*[right=CDecStep]
  {
    \kappa(a)=\ClassDef(c_{opt},\rho,S)
    \\
    \sigma,\kappa, \CV, H\stepHash{\alpha} H'
  }
  {
    \sigma,\kappa,\CV, \mathtt{class}\
  \{a: H\}\stepStmt{\alpha} \sigma,\kappa,\CV,
  \mathtt{class}\
  \{a: H'\}
  }
  \and
  \inferrule*[right=CDecU]
  {
    \kappa(a)=\ClassDef(\bot,\rho,s)
    \\
    s'=\merge(\rho, \HV)
    \\
    \beta \baseof{\kappa} \alpha
  }
  {
    \sigma,\kappa,\CV, \mathtt{class}\ \{a: \HV\}\stepStmt{\alpha} \sigma,\defupdate{\kappa}{a}{\Declared(\beta)},\CV, \mathtt{scope}\ (\classScope{a})\ \{s' \seq s\}
  }
  \and
 %
  \and
  \inferrule*[right=CDecPU]
  {
    \kappa(a)=\ClassDef(b,\rho,s)
    \\
    \kappa(b)=\ClassDef(c_{opt},\rho',s')
  }
  {
    \sigma,\kappa,\CV, \mathtt{class}\
  \{a: \HV\}\stepStmt{\alpha}\sigma,\kappa,\CV,
  \mathtt{include}\ b\seq \mathtt{class}\
  \{a: \HV\}
  }
\and
  \inferrule*[right=CDecPD]
  {
    \kappa(a)=\ClassDef(b,\rho,s)
    \\
    \kappa(b)=\Declared(\beta)
    \\
    s'=\merge(\rho,\HV)
  }
  {
    \sigma,\kappa,\CV, \mathtt{class}\ \{a: \HV\}
    \stepStmt{\alpha}
    \sigma, \defupdate{\kappa}{a}{\Declared(\classScope{b})},\CV, \mathtt{scope}\ (\classScope{a})\ \{s' \seq s\}
  }
\end{mathpar}

%% file: rule/manifest/statement.tex
\begin{mathpar}
  \inferrule*[Right=TopScope]
  {
    \sigma,\kappa,\CV , s \stepStmt{\topScope}
    \sigma',\kappa',\CV' , s'
  }
  {
    \sigma,\kappa,\CV , s \stepManifest{N} \sigma',\kappa',\CV'
    , s'
  }
\end{mathpar}

%% file: rule/manifest/node-definition.tex
\begin{mathpar}
  \inferrule*[right=NodeMatch]
  {
    \nodeMatch(N,Q)
  }
  {
    \sigma,\kappa,\CV ,  \mathtt{node}\; Q\;\{ s
    \}\stepManifest{N} \sigma,\kappa,\CV , \mathtt{scope}\ (\nodeScope)\ s
  }
  \and
  \inferrule*[right=NodeNoMatch]
  {
    \neg\nodeMatch(N,Q)
  }
  {
    \sigma,\kappa,\CV ,  \mathtt{node}\; Q\;\{ s
    \}\stepManifest{N} \sigma,\kappa,\CV ,  \kskip
  }
\end{mathpar}

%% file: rule/manifest/defined-resource-type.tex
\begin{mathpar}
  \inferrule*[right=RDef]
  {
    u \notin \dom(\kappa)
  }
  {
    \sigma,\kappa,\CV,\mathtt{define} \ u\ (\rho)\ \{s\}
    \stepManifest{N} 
    \sigma, \defupdate{\kappa}{u}{\ResourceDef(\rho,s)},\CV,\mathtt{skip}
  }
\end{mathpar}

%% file: rule/manifest/class-definition.tex
\begin{mathpar}
  \inferrule*[right=CDef]
  {
    a \notin \dom(\kappa)
  }
  {
    \sigma,\kappa,\CV, \mathtt{class}\ a\ \{s\}
    \stepManifest{N}
    \sigma, \defupdate{\kappa}{a}{\ClassDef(\bot,\varepsilon,s)},\CV, \mathtt{skip}
  }
  \and
  \inferrule*[right=CDefI]
  {
    a \notin \dom(\kappa)
  }
  {
    \sigma,\kappa,\CV, \mathtt{class}\ a\ \mathtt{inherits}\ b\ \{s\}
    \stepManifest{N}
    \sigma, \defupdate{\kappa}{a}{\ClassDef(b,\varepsilon,s)},\CV, \mathtt{skip}
  }
  \and
  \inferrule*[right=CDefP]
  {
    a \notin \dom(\kappa)
  }
  {
    \begin{array}{l}\sigma,\kappa,\CV, \mathtt{class}\ a\ (\rho)\ \{s\}
    \stepManifest{N}
    \sigma, \defupdate{\kappa}{a}{\ClassDef(\bot,\rho,s)},\CV, \mathtt{skip}\end{array}
  }
  \and
  \inferrule*[right=CDefPI]
  {
    a \notin \dom(\kappa)
  }
  {
    \sigma,\kappa,\CV, \mathtt{class}\ a\ (\rho)\ \mathtt{inherits}\ b\ \{s\}
    \stepManifest{N}
    \sigma, \defupdate{\kappa}{a}{\ClassDef(b,\rho,s)},\CV, \mathtt{skip}
  }
\end{mathpar}

%% file: sec/metatheory.tex
\section{Metatheory}
\label{sec:metatheory}

Because Puppet has not been designed with formal properties in mind, there is
relatively little we can say formally about the ``correctness'' of \muPuppet.
Instead, the main measure of correctness is the degree to which \muPuppet
agrees with the behaviour of the main Puppet implementation, which is the
topic of the next section.  Here, we summarise two properties of \muPuppet
that guided our design of the rules, and provide some justification
for the claim that \muPuppet is `declarative'.  First, evaluation is deterministic: a
given manifest can evaluate in at most one way.

\begin{theorem}
  [Determinism]
All of the evaluation relations of \muPuppet are deterministic: 
\begin{itemize}
\item If $\sigma, \kappa, \CV,e \stepExpr{\alpha} e'$ and $\sigma,
  \kappa, \CV,e \stepExpr{\alpha} e''$ then $e' = e''$.
  \item If $\sigma,\kappa,\CV, s \stepStmt{\alpha}
    \sigma',\kappa',\CV' , s'$ and $\sigma,\kappa,\CV, s \stepStmt{\alpha}
    \sigma'',\kappa'',\CV'' , s''$ then $\sigma'=\sigma''$, $\kappa' =
    \kappa''$, $\CV' = \CV''$ and $s' = s''$.
  \item If
    $\sigma,\kappa,\CV, s \stepManifest{N} \sigma',\kappa',\CV' , m'$
    and
    $\sigma,\kappa,\CV, m \stepManifest{N} \sigma'',\kappa'',\CV'' ,
    m''$
    then $\sigma'=\sigma''$, $\kappa' = \kappa''$, $\CV' = \CV''$ and
    $m' = m''$.
\end{itemize}
\end{theorem}
\begin{proof}
  Straightforward by induction on derivations.
\end{proof}

Second, in \muPuppet, evaluation is monotonic in the sense that:
\begin{itemize}
\item Once a variable binding is defined in $\sigma$, its value never changes,
  and it remains bound until the end of the scope in which it was bound.
\item Once a class or resource definition is defined in $\kappa$, its
  definition never changes, except that a class's definition may
  change from $\ClassDef(c_{opt},\rho,s)$ to $\Declared(\beta)$.
\item Once a resource is declared in $\CV$, its title, properties and
  values never change.
\end{itemize}
We can formalise this as follows.
\begin{definition}
  We define orderings $\sqsubseteq$ on variable environments,
  definition environments and catalogs as follows:
  \begin{itemize}
  \item $\sigma \sqsubseteq \sigma'$ when  $x \in dom(\sigma_\alpha)$ implies
    that either $\sigma_\alpha(x) =
    \sigma'_\alpha(x)$ or $\alpha = \defScope{\beta}$ for some $\beta$
    and $x \not\in dom(\sigma'_\alpha)$.
\item $\kappa \sqsubseteq \kappa'$ when $a \in dom(\kappa)$ implies
  either 
  $\kappa(a) = \kappa'(a)$ or $\kappa(a) = \ClassDef(c_{opt},\rho,s)$
  and $\kappa'(a) = \Declared(\beta)$.
\item $\CV \sqsubseteq \CV'$ when there exists $\CV''$ such that $\CV\visiblespace
  \CV'' = \CV'$.
\item $(\sigma,\kappa,\CV) \sqsubseteq (\sigma',\kappa',\CV')$ when
  $\sigma \sqsubseteq \sigma'$, $\kappa \sqsubseteq \kappa'$ and $\CV
  \sqsubseteq \CV'$.
  \end{itemize}
\end{definition}
\begin{theorem}
  [Monotonicity]
The statement and manifest evaluation relations of \muPuppet are
monotonic in $\sigma,\kappa,\CV$:
\begin{itemize}
  \item If $\sigma,\kappa,\CV, s \stepStmt{\alpha}
    \sigma',\kappa',\CV' , s'$ then $(\sigma,\kappa,\CV) \sqsubseteq
    (\sigma',\kappa',\CV')$.
  \item If
    $\sigma,\kappa,\CV, s \stepManifest{N} \sigma',\kappa',\CV' , m'$
then $(\sigma,\kappa,\CV) \sqsubseteq
    (\sigma',\kappa',\CV')$.
\end{itemize}
\end{theorem}
\begin{proof}
  Straightforward by induction.  The only interesting cases are the
  rules in which $\sigma$, $\kappa$ or $\CV$ change; in each
  case the conclusion is immediate.
\end{proof}

These properties are not especially surprising or difficult to prove;
nevertheless, they provide some justification for calling
\muPuppet a `declarative' language.  However, \muPuppet does not satisfy some
other desirable properties.  For example, as we have seen, the order in which
variable definitions or resource or class declarations appear can affect the
final result.  Likewise, there is no notion of `well-formedness' that
guarantees that a
\muPuppet program terminates successfully: compilation may diverge or
encounter a run-time error.  Furthermore, full Puppet does not satisfy
monotonicity, because of other non-declarative features that we have chosen
not to model.  Further work is needed to identify and prove desirable
properties of the full Puppet language, and identify subsets of (or
modifications to) Puppet that make such properties valid.

%% file: sec/implementation.tex
\section{Implementation and Evaluation}
\label{sec:impl}

We implemented a prototype parser and evaluator \muPuppet in Haskell (GHC
8.0.1).  The parser accepts source syntax for features of
\muPuppet as described in the Puppet documentation and produces
abstract syntax trees as described in Section~\ref{sec:semantics}. The
evaluator implements \muPuppet compilation based on the rules shown in
\appref{rules}.  The implementation constitutes roughly 1300 lines of Haskell
code.  The evaluator itself is roughly 400 lines of code, most of which are
line-by-line translations of the inference rules.
  
We also implemented a test framework that creates an Ubuntu 16.04.1 (x86\_64)
virtual machine with Puppet installed, and scripts which run each example
using both \muPuppet and Puppet and compare the resulting messages and catalog
output.

\subsection{Test cases and results}

During our early investigations with Puppet, we constructed a test set of 52
manifests that illustrate Puppet's more unusual features, including resources,
classes, inheritance, and resource type definitions.  The tests include
successful examples (where Puppet produces a catalog) and failing examples
(where Puppet fails with an error); we found both kinds of tests valuable for
understanding what is possible in cases where the documentation was
unspecific.

\input{fig/results-summary.tex}

We used these test cases to guide the design of \muPuppet, and developed 16
additional test cases along the way to test corner cases or clarify behaviour
that our rules did not originally capture correctly.  We developed further
tests during debugging and to check the behaviour of Puppet's (relatively)
standard features, such as conditionals and case statements, arrays, and
hashes.  We did not encounter any surprises there so we do not present these
results in detail.

We summarise the test cases and their results in Table~\ref{tab:summary}.  The
``Feature'' column describes the classification of features present in our
test set.  The ``\#Tests'' and ``\#Pass'' columns show the number of tests in
each category and the number of them that pass.  A test that is intended to
succeed passes if both Puppet and
\muPuppet succeed and produce the same catalog (up to reordering of
resources); a test that is intended to fail passes if both Puppet and
\muPuppet fail.  The ``\#Unsupported'' column shows the number of test
cases that involve features that \muPuppet does not handle.
All of the tests either pass or use features that are not supported
by \muPuppet.   Features that \muPuppet (by design) does not support
are italicised.  

All of the examples listed in the above table are included in the
supplementary material, together with the resulting catalogs and error
messages provided by Puppet.

\subsection{Other Puppet examples}

A natural source of test cases is Puppet's own test suite or, more
generally, other Puppet examples in public repositories.  Puppet does
have a test suite, but it is mostly written in Ruby to test internal
functionality.  We could find only 43 Puppet language tests in the
Puppet repository on
GitHub\footnote{https://github.com/puppetlabs/puppet/tree/master/spec/fixtures/unit/parser/lexer}.
These tests appear to be aimed at testing parsing and lexing
functionality; they are not accompanied by descriptions of the desired
catalog result.  Some of the tests also appear to be out of date: five
fail in Puppet 4.8.  Of the remaining test cases that Puppet 4.8 can
run, 20 run correctly in \muPuppet (with minor modifications) while 18
use features not yet implemented in \muPuppet.

We also considered harvesting realistic Puppet configurations from
other public repositories; however, this is not straightforward since
real configurations typically include confidential or
security-critical parameters so are not publicly available.  An
alternative would be to harvest Puppet modules from publicly available
sources such as PuppetForge, which often include test manifests to
show typical usage.  However, these test cases usually do not come with sample
results; they are mainly intended for illustration.  

We examined the top 10 Puppet modules (apache, ant, concat, firewall,
java, mysql, ntp, postgresql, puppetdb, and stdlib) on the official
PuppetForge module site and searched for keywords and other symbols in
the source code to estimate the number of uses of Puppet features such
as classes, inheritance, definitions, resource collectors/virtual
resources, and ordering constraints.  Classes occurred in almost all
modules, with over 200 uses overall, and over 50 uses of inheritance.
Resource type definitions were less frequent, with only around 40
uses, while uses of resource collectors and virtual resources were rare: there
were only 10 uses overall, distributed among 5 packages.  Ordering
constraints were widely used, with over 90 occurrences in 8 packages.
Due to the widespread use of ordering constraints, as well as other
issues such as the lack of support for general strings and string
interpolation in \muPuppet, we were not able to run \muPuppet on these
Puppet modules. This investigation suggests that to develop tools or
analyses for real Puppet modules based on \muPuppet will require both
conceptual steps (modelling ordering constraints and non-declarative
features such as resource collectors) as well as engineering effort
(e.g. to handle Puppet's full, idiosyncratic string interpolation
syntax).

\subsection{Unsupported features}
\label{sec:implementation:discussion}

Our formalisation handles some but not all of the distinctive features of
Puppet. As mentioned in the introduction, we chose to focus effort on the
well-established features of Puppet that appear closest to its declarative
aspirations.  In this section we discuss the features we chose not to support
and how they might be supported in the future, in increasing order of
complexity.

\Paragraph{String interpolation} Puppet supports a rich set of
string operations including string interpolation (i.e. variables and other
expression forms embedded in strings).  For example, writing
\verb|"Hello ${planet['earth']}!"|
produces \verb|"Hello world!"| if variable \verb|planet| is a hash
whose \verb|'earth'| key is bound to \verb|'World'|. String
interpolation is not conceptually difficult but it is widely used and
desugaring it correctly to plain string append operations is an engineering challenge.

\Paragraph{Dynamic data types} Puppet 4 also supports type annotations, which
are checked dynamically and can be used for automatic validation of
parameters. For example, writing \verb|Integer $x = 5| in a parameter list
says that $x$ is required to be an integer and its default value is 5.  Types
can also express constraints on the allowed values: for example,
\verb|5 =~ Integer[1,10]| is a valid expression that evaluates to
\verb|true| because 5 is an integer between 1 and 10.  Data types are
themselves values and there is a type \verb|Type| of data types.

\Paragraph{Undefined values and strict mode} By default, Puppet treats
an undefined variable as having a special ``undefined value''
\verb|undef|.  Puppet provides a ``strict'' mode that treats an
attempt to dereference an undefined variable as an error.  We have focused on
modelling strict semantics, so our rules get stuck if an attempt is made to
dereference an undefined variable; handling explicit undefined values seems
straightforward, by changing the definitions of lookup and related operations
to return \verb|undef| instead of failing.

\Paragraph{Functions, iteration and lambdas} As of version 4, Puppet
allows function definitions to be written in Puppet and also includes
support for iteration functions (\verb|each|, \verb|slice|,
\verb|filter|, \verb|map|, \verb|reduce|, \verb|with|) which take
lambda blocks as arguments.  The latter can only be used as function
arguments, and cannot be assigned to variables, so Puppet does not yet have
true first-class functions.  We do see no immediate obstacle to handling these
features, using standard techniques.

\Paragraph{Nested constructs and multiple definitions} 
We chose to consider only top-level definitions of classes and defined
resources, but Puppet allows nesting of these constructs, which also makes it
possible for classes to be defined more than once.  For example:
\begin{lstlisting}
class a {
  $x1 = "a"
  class b {
    $y1 = "b"
  }
}

class a::b {
  $y2 = "ab"
}
include a
include a::b
\end{lstlisting}
Surprisingly, \emph{both} line 4 and line 9 are executed (in
unspecified order) when \verb|a::b| is declared, so both
\verb|$::a::b::y1| and \verb|$::a::b::y2| are in scope at the
end.  Our impression is that it would be better to simply reject Puppet
manifests that employ either nested classes or multiple definitions, since
nesting of class and resource definitions is explicitly discouraged by the
Puppet documentation.

\Paragraph{Dynamically-scoped resource defaults}
Puppet also allows setting resource defaults. For example one can
write (using the capitalised resource type \verb|File|):
\begin{lstlisting}
  File { owner => "alice" }
\end{lstlisting}
to indicate that the default owner of all files is \verb|alice|. Defaults can
be declared in classes, but unlike variables, resourced defaults are
dynamically scoped; for this reason, the documentation and some authors both
recommend using resource defaults sparingly.  Puppet 4 provides an alternative
way to specify defaults as part of the resource declaration.

\Paragraph{Resource extension and overriding}
In Puppet, attributes can be added to a resource which has been previously
defined by using a \emph{reference} to the resource, or removed by setting
them to \verb|undef|.
\begin{lstlisting}
class main {
  file { "file": owner => "alice" }
  File["file"] { mode => "0755" }
}
\end{lstlisting}
However, it is an error to attempt to change the value of an
already-defined resource, unless the updating operation is performed
in a \emph{subclass} of the class in which the resource was originally
declared.  For example:
\begin{lstlisting}
class main::parent {
  file { "file":
    owner => "bob",
    source => "the source"
  }
}
class main inherits main::parent {
  File["file"] {
    owner => "alice",
    source => undef
  }
}
\end{lstlisting}
This illustrates that code in the derived class is given special permission to
override any resource attributes that were set in the base class.  Handling
this behaviour seems to require (at least) tracking the classes in which
resources are declared.

\Paragraph{Resource collectors and virtual resources}
\emph{Resource collectors} allow for
selecting, and also updating, groups of resources specified via predicates.
For example, the following code declares a resource and then immediately uses
the collector
\verb"File <|title == "file"|>"
to modify its parameters.  
\begin{lstlisting}
class main {
	file { "file": owner => "alice" }
	File <| title == "file" |> {
		owner => "bob",
		group => "the group",
	}  
}
\end{lstlisting}
Updates based on resource collectors are unrestricted, and the scope of the
modification is also unrestricted: so for example the resource collector
\verb"File<|owner='root'|>" will select all files owned by root, and
potentially update their parameters in arbitrary ways.  The Puppet
documentation recommends using resource collectors only in idiomatic ways,
e.g. using the title of a known resource as part of the predicate. Puppet also
supports \emph{virtual resources}, that is, resources with parameter values
that are not added to the catalog until declared or referenced elsewhere.
Virtual resources allow a resource to be declared in one place without the
resource being included in the catalog. The resource can then be {\em
realised} in one or more other places to include it in the catalog. Notice
that you can realise virtual resources before declaring them:
\begin{lstlisting}
class main {
	realize User["alice"]
	@user { "alice": uid => 100 }
	@user { "bob": uid => 101 }
	realize User["alice"]
}
\end{lstlisting}

 As Shambaugh et
al.~\cite{shambaugh16pldi} observe, these features can have global
side-effects and make separate compilation impossible; the Puppet
documentation also recommends avoiding them if possible.  We have not
attempted to model these features formally, and doing so appears to be a
challenging future direction.

\Paragraph{Ordering constraints} By default, Puppet does not guarantee
to process the resources in the catalog in a fixed order.  To provide
finer-grained (and arguably more declarative) control over ordering,
Puppet provides several features: special \emph{metaparameters} such
as \verb|ensure|, \verb|require|, \verb|notify|, and \verb|subscribe|,
\emph{chaining arrows} \verb|->| and \verb|~>| that declare dependencies
among resources, and the \emph{require function} that includes a class
and creates dependencies on its resources.  From the point of
view of our semantics, all of these amount to ways to define
dependency edges among resources, making the catalog into a
\emph{resource graph}. Puppet represents the chaining arrow dependencies using 
metaparameters, so we believe this behaviour can be handled using techniques
similar to those for resource parameter overrides or resource collectors.  The
rules for translating the different ordering constraints to resource graph
edges can be expressed using Datalog rules~\cite{shambaugh16pldi} and this
approach may be adaptable to our semantics too.

%% file: fig/results-summary.tex
\begin{table}[tb]
  \centering
\begin{tabular}{lccc}
Feature & \#Tests & \#Pass & \#Unsupported\\
\hline
Statements & 11 & 11 & 0 \\ 
~~Assignment & 2 & 2 & 0 \\ 
~~Case & 1 & 1 & 0 \\ 
~~If & 4 & 4 & 0 \\ 
~~Unless & 4 & 4 & 0 \\ 
\hline
Resources & 18 & 11 & 7 \\ 
~~Basics & 2 & 2 & 0 \\ 
~~Variables & 3 & 3 & 0 \\
~~User defined resource types & 5 & 5 & 0 \\ 
~~\textit{Virtual resources} & 1 & 0 & 1 \\ 
~~\textit{Default values} & 1 & 0 & 1 \\ 
~~\textit{Resource extension} & 4 & 0 & 4\\ 
~~\textit{Ordering Constraints} & 2 & 1 & 1 \\ 
\hline
Classes & 32 & 22 & 10 \\  
~~Basics & 4 & 4 & 0 \\ 
~~Inheritance & 3 & 3 & 0 \\ 
~~Scope & 2 & 2 & 0 \\ 
~~Variables \& classes & 6 & 6 & 0 \\ 
~~Class Parameters & 6 & 6 & 0 \\ 
~~\textit{Overriding} & 5 & 0 & 5 \\ 
~~\textit{Nesting and redefinition} & 6 & 1 & 5 \\ 
\hline
Nodes & 8 & 8 & 0 \\ 

\hline
\textit{Resource Collectors} & 9 & 0 & 9 \\ 
~~\textit{Basics} & 1 & 0 & 1 \\ 
~~\textit{Collectors, references \& variables} & 3 & 0 & 3 \\ 
~~\textit{Application order} & 5 & 0 & 5 \\ 
\hline
\end{tabular}

\caption{Summary of test cases. Features in \textit{italics} are not
  supported in \muPuppet.}
\label{tab:summary}
\end{table}

%% file: sec/related-work.tex
\section{Related work}
\label{sec:related}

Other declarative configuration frameworks include LCFG~\cite{lcfgbook}, a
configuration management system for Unix, and
SmartFrog~\cite{goldsack09sigops}, a configuration language developed by HP
Labs. Of these, only SmartFrog has been formally specified; Herry and
Anderson~\cite{DBLP:journals/jnsm/AndersonH16} propose a formal semantics and
identify some complications, including potential termination problems
exhibited by the SmartFrog interpreter. Their semantics is presented in a
denotational style, in contrast to the small-step operational semantics
presented here for Puppet. Other systems, such as Ponder~\cite{damianou02phd},
adopt an operational approach to policies for distributed systems.

Beyond this, there are relatively few formal studies of configuration
languages, and we are aware of only two papers on Puppet specifically.
Vanbrabant \etal~\cite{vanbrabant11lisa} propose an access control technique
for an early version of Puppet based on checking whether the changes to the
catalog resulting from a change to the manifest are allowed by a policy.
Catalogs are represented as XML files and allowed changes are described using
path expressions.  Shambaugh
\etal~\cite{shambaugh16pldi} present a configuration verification tool
for Puppet called Rehearsal. Their tool is concerned primarily with the
``realisation'' stage of a Puppet configuration, and focuses on the problem of
determinacy analysis: that is, determining whether a proposed reconfiguration
leads to a unique result state. Shambaugh
\etal consider a subset of Puppet as a source language,
including resources, defined resources, and dependencies. However,
some important subtleties of the semantics were not investigated.
Compilation of definitions and ordering
constraints was described at a high level but not formalised; classes and
inheritance were not mentioned, although their implementation handles
simple cases of these constructs.
Our work complements Rehearsal: Rehearsal analyses the determinacy of the realisation
stage, while our work improves understanding of the compilation stage.

The present work continues a line of recent efforts to study the semantics of
programming and scripting languages ``in the wild''. There have been efforts
to define semantics for JavaScript~\cite{maffeis08aplas,guha10ecoop},
R~\cite{morandat12ecoop}, PHP~\cite{filaretti14ecoop}, and
Python~\cite{politz13oopsla}.  Work on formal techniques for
Ruby~\cite{ueno14aplas} may be especially relevant to Puppet: Puppet is
implemented in Ruby, and plugins can be written in Ruby, so modelling the
behaviour of Puppet as a whole may require modelling both the Puppet
configuration language and the Ruby code used to implement plugins, as well as
other tools such as Hiera\footnote{\url{https://docs.puppet.com/hiera/}.} that
are an increasingly important component of the Puppet toolchain. However,
Puppet itself differs significantly from Ruby, and Puppet ``classes'' in
particular bear little relation to classes in Ruby or other object-oriented
languages.

%% file: sec/conclusion.tex
\section{Conclusions}
\label{sec:concl}

Rigorous foundations for configuration frameworks are needed to improve the
reliability of configurations for critical systems. Puppet is a popular
configuration framework, and is already being used in safety-critical domains
such as air traffic control.\footnote{\url{https://archive.fosdem.org/2011/schedule
/event/puppetairtraffic.html}.}

Even if each individual component of such a system is formally verified,
misconfiguration errors can still lead to failures or vulnerabilities, and the
use of these tools at scale means that the consequences of failure are also
potentially large-scale. The main contribution of this paper is an operational
semantics for a subset of Puppet, called \muPuppet, that covers the
distinctive features of Puppet that are used in most Puppet configurations,
including resource, node, class, and defined resource constructs. Our rules
also model Puppet's idiosyncratic treatment of classes, scope, and
inheritance, including the dynamic treatment of node scope.

We presented some simple metatheoretic properties that justify our
characterisation of \muPuppet as a `declarative' subset of Puppet, and we
compared \muPuppet with the Puppet 4.8 implementation on a number of examples.
We also identified idiosyncrasies concerning evaluation order and scope where
our initial approach differed from Puppet's actual behaviour. Because Puppet
is a work in progress, we hope that these observations will contribute to the
evolution and improvement of the Puppet language. In future work, we plan to
investigate more advanced features of Puppet and develop semantics-based
analysis and debugging techniques; two natural directions for future work are
investigating Puppet's recently-added type system, and developing
\emph{provenance} techniques that can help explain where a catalog
value or resource came from, why it was declared, or why manifest
compilation failed~\cite{anderson12tapp}.

%% file: acks.tex
\paragraph*{Acknowledgments}
Fu was supported by a Microsoft Research PhD studentship. Perera and
Cheney were supported by the Air Force Office of Scientific Research,
Air Force Material Command, USAF, under grant number
FA8655-13-1-3006. The U.S. Government and University of Edinburgh are
authorised to reproduce and distribute reprints for their purposes
notwithstanding any copyright notation thereon. Perera was also
supported by UK EPSRC project EP/K034413/1.  We also gratefully
acknowledge Arjun Guha for comments on an early version of this paper
and Henrik Lindberg for discussions about Puppet's semantics and
tests.

%% file: appendix.tex
\section{Glossary}
\label{app:glossary}

In this section we provide short definitions of certain terms as they
are used in Puppet, for reference while reading the paper.

\begin{description}
\item[agent] Client software (Puppet agent) running on each node, responsible
for liaising with the Puppet master.
\item[catalog] A collection of \textsl{resources}, together with \textsl{ordering
  constraints} forming a directed acyclic graph among the resources.
\item[class] A named collection of resources.  Classes can take
  parameters and their contents can be processed as a result of an
  \textsl{include-like declaration} (\verb|include a|) or a
  \textsl{resource-like declaration} (\verb|class { a : ...}|).
  Declaring a class more than once has no effect; however, once a
  class has been declared with parameters (using a resource-like
  declaration), it cannot be redeclared (with possibly different parameters) again.
\item[compile] To process a \textsl{manifest} in order to construct the
  \textsl{catalog} for a given \textsl{node}. Compilation in Puppet is closer
  is what is usually called \emph{evaluation} in programming languages
  terminology, because an expression-like form (the manifest) is
  ``normalised'' into a value-like form (the catalog).
\item[compilation error] An error that arises during the compilation of a
Puppet manifest; a compilation error in Puppet is what would normally be
described as a \emph{runtime} error, since only errors that arise in code that
is actually executed are reported.
\item[data type] An expression such as \verb|Integer| or
  \verb|Integer[1,10]| that can be used for run-time type checking or
  parameter validation.
\item[declare] To assert the existence of, and describe the parameters
  of, a resource.  When a built-in resource is declared, it is added to
  the \textsl{catalog}.  When a user-defined resource or class is declared, its
  body is processed.
\item[define] The definition of class or defined resource type is the
  program construct that names, lists the parameters, and gives the
  body of the construct.
\item[environment] Group of Puppet nodes that share configuration settings.
Useful for testing.
\item[fact(s)] Information about the \textsl{node} that the
  \textsl{agent} collects, which is passed to the manifest compiler
  as the value of a special hash-valued variable \verb|$facts|.
\item[hash] A dictionary consisting of key-value pair bindings; keys
  may be `scalar' values such as integers and strings.
\item[include-like declaration] Syntax for \textsl{declaring} a \textsl{class}
  using a special function such as \verb|include|, \verb|contain|,
  \verb|require|, or \verb|hiera_include|.
\item[manifest] The Puppet source code file(s) containing code in the
  Puppet configuration language, which is \textsl{compiled} to produce
  a \textsl{catalog} for a given \textsl{node}.
\item[node] A computer system (physical or virtual) that is managed by
  Puppet, identified by its hostname. Communicates with Puppet master via
Puppet agent software.
\item[node scope] The scope of the active node definition, typically
containing node-specific overrides of global definitions.
\item[ordering constraints] Constaints on the order in which resources
  are processed.  These constraints can be defined explicitly using
  \verb|->|, or implicitly using \textsl{metaparameters} such as
  \verb|require| and \verb|before|.  Constraints can also be annotated
  as conveying information among resources, using \verb|~>| or the
  \textsl{metaparameters} \verb|notify| and \verb|subscribe|.
\item[package] A particular Puppet resource type, representing a package on a
  \textsl{node}, abstracting over the particular operating system or
  package manager in use on that system.  
\item[resource] A representation of a system component of a \textsl{node}
  that is managed by Puppet; for example, a user account, file, package, or
  webserver.  A resource has a \textsl{type} (such as \verb|file| or
  \verb|user|) and \textsl{title} string.  A \textsl{catalog} can contain at most one
  resource with a given type and title.
\item[resource collector] An expression \verb!T<|pred|>! that collects all resources of
  a given type \verb|T| whose parameters satisfy a predicate \verb|pred|.  Resource
  collectors can also be used to override resource parameter values.
  Sometimes called the \textsl{spaceship operator}.
\item[resource-like declaration] Syntax for declaring a
  \textsl{class} that resembles the syntax used to declare resources.
  Resource-like declarations are the only way to override the default
  values of parameters or provide values for parameters that have no
  defaults.
\item[resource type] The type of a resource. Puppet \puppetversion also contains
  \textsl{data types} that classify values of expressions.
\item[top scope] The root namespace of a Puppet manifests; parent of any node
scope.
\end{description}

\section{Features supported}
\label{app:features}
\input{fig/features-supported}

Figure~\ref{fig:language-coverage} summarises our coverage of the Puppet \puppetversion
language.

\section{Operational semantics}
\label{app:rules}
\subsection{Environment operations}

\begin{itemize}
\item The operation $\lookup$ returns the first value associated with a
  parameter in a hash, or $\bot$ if the parameter is not a key in the
  hash.  ($Scalar$ is the set of possible hash keys, i.e. integers or strings.)
\input{rule/environment/lookup}

\item The function $\lookupCat$ returns an attribute value of a resource value
in a catalog $\CV$ if the resource has the type $t$, the title $w$ and the
attribute $x$. It returns $\bot$ if there is no such a resource value.

\input{rule/environment/lookupCatalog}

\item The partial function $\merge$ takes a list of parameters with optional
  default expressions, and a hash mapping parameter names to overriding
  values, and returns a statement initialising each parameter to its default
  or overridden value.

\input{rule/environment/merge}

\noindent Note that $\merge((\$x, \rho), \HV)$ is undefined unless $x$ is a key of $\HV$.

\item Update $\fupdate{\sigma}{\alpha}{x}{v}$ add a new
  variable-value pair to a variable environment. It is defined as follows:
\input{rule/environment/update}
where $\Set{Env} = \Set{Scope} \times \Set{Var} \to \Set{Value}$.

\item Specialisation $\sigma_\alpha$ is the variable environment $\sigma$
specialised to a particular scope $\alpha$, defined as follows:
\input{rule/environment/specialize}

\item Clearing $\clear(\sigma,\alpha)$ makes all variables
  in scope $\alpha$ undefined, leaving all other variable bindings
  unchanged.  It is used to clean up $\defScope{\alpha}$ scopes at the
  end of their lifetime.  This operation is defined as follows:
\input{rule/environment/clear}

\item The \emph{base scope} of a scope $\alpha$ in the context of a
  definition environment $\kappa$ is defined as follows:

\input{rule/environment/baseof}
\item The parent of a scope $\alpha$ in the context of a definition environment
$\kappa$ is defined as follows:

\input{rule/environment/parentof}
\end{itemize}

\subsection{Expressions ($\sigma, \kappa , \CV,e \stepExpr{\alpha} e'$)}

\subsubsection{Variables}

\input{rule/expression/variable}
\subsubsection{Arithmetic expressions}

\input{rule/expression/arithmetic}

Here, $+_\mathbb{Z}$ stands for the usual integer addition function. The
evaluation rules for $-,\ast,/$ are similar to those for $+$, and omitted.

\subsubsection{Comparison expressions}

\input{rule/expression/comparison}

The rules for other comparison operators are similar and omitted.

\subsubsection{Boolean expressions}

\input{rule/expression/boolean}

The rules for disjunction are analogous and are omitted.

\subsubsection{Array}

\input{rule/expression/array}

\subsubsection{Hash}

\input{rule/expression/hash}

\subsubsection{Selectors}

The predicate $\caseMatch$ abstracts over the details of matching values
against selector cases. Real Puppet checks any \kdefault clause last, failing
with an exception if there is no case which matches; for simplicity we omit
these details from the formalism.
 
\input{rule/expression/selector}

\subsubsection{Dereference}

\input{rule/expression/dereference}
\subsection{Resources ($\sigma,\kappa , H
  \stepHash{\alpha}H' $ and $\sigma,\kappa , e:H \stepRes{\alpha}
  e':H'$)}

\input{rule/resource}

\subsection{Statements ($\sigma,\kappa,\CV , s
  \stepStmt{\alpha} \sigma',\kappa',\CV' , s'$)}

\subsubsection{Expression statements}

An expression can occur as a statement. Its value is ignored.

\input{rule/statement/expression}

\subsubsection{Sequential composition}

\input{rule/statement/sequential-composition}

\subsubsection{Assignment}


\input{rule/statement/assignment}
\subsubsection{If}

\input{rule/statement/if}

\subsubsection{Unless}

\input{rule/statement/unless}

\subsubsection{Case}

As with selectors, the predicate $\caseMatch$ abstracts over the details of
pattern-matching. Again, in real Puppet the \kdefault case is always checked
last, regardless of the order of the cases.

\input{rule/statement/case}

\subsubsection{Resource declarations}

\input{rule/statement/resource-declaration}

\subsubsection{Defined resource types}

\input{rule/statement/defined-resource-type}
\subsubsection{Include}

\input{rule/statement/include}
\subsubsection{Resource-like class declarations}

\input{rule/statement/resource-like-class-declaration}
\subsubsection{Scope}

\input{rule/statement/scope}
\subsection{Manifests ($\sigma,\kappa,\CV , m \stepManifest{N} \sigma',\kappa',\CV' ,
m'$)}

\subsubsection{Top-Level Statements}

\input{rule/manifest/statement}
\subsubsection{Sequential Composition}

\input{rule/manifest/sequential-composition}

\subsubsection{Node Definitions}

The predicate $\nodeMatch$ abstracts over the details of matching values
against node specifications.

\input{rule/manifest/node-definition}
\subsubsection{Defined resource types}

\input{rule/manifest/defined-resource-type}
\subsubsection{Class Definitions}

\input{rule/manifest/class-definition}

%% file: fig/features-supported.tex
\newcommand*{\twoColumnHeading}[1]{\multicolumn{2}{c}{\textbf{\emph{#1}}}}
\begin{figure}
\centering
\begin{tabular}{lc}
Puppet \puppetversion feature
&
Modelled?
\\
\twoColumnHeading{Built-in operators}
\\
Comparison operators (\tt{==}, \tt{!=}, \tt{>=}, \tt{<=}, \tt{>}, and \tt{<})
&
\Yes
\\
Boolean operators ({\tt{and}, \tt{or}, \tt{!}})
&
\Yes
\\
Arithmetic operators (\tt{+}, \tt{-}, \tt{*}, \tt{/}, \tt{\%}, \tt{<}\tt{<}, \tt{>}\tt{>})
&
\Yes
\\
Array and hash operators (\tt{*}, \tt{<}\tt{<}, \tt{+}, \tt{-})
&
\No
\\
Assignment (\tt{=})
&
\Yes
\\
\twoColumnHeading{Other core features}
\\
Conditional forms (\tt{if}, \tt{unless}, \tt{case}, \tt{?})
&
\Yes
\\
Node definitions (\tt{node})
&
\Yes
\\
Comments
&
\Yes
\\
Facts and pre-set variables
&
\No
\\
\twoColumnHeading{Data types}
\\
Strings, Numbers \& Booleans
&
\Yes
\\
Arrays
&
\Yes
\\
Hashes
&
\Yes
\\
Regular expressions
&
\No
\\
\tt{Sensitive}
&
\No
\\
\tt{undef}
&
\No
\\
Resource references
&
\Yes
\\
Resource data types
&
\No
\\
\tt{default}
&
\Yes
\\
Data type annotations and values of type \tt{Type}
&
\No
\\
\twoColumnHeading{Resources}
\\
Built-in resource types
&
\Yes
\\
Defined resource types (\tt{define})
&
\Yes
\\
Multiple resource bodies; per-expression default attributes
&
\No
\\
Set attributes from hash (\tt{* =>})
&
\No
\\
Abstract resource types (\tt{Resource[\ldots]})
&
\No
\\
Multiple resource titles
&
\No
\\
Add attributes to existing resources
&
\No
\\
\twoColumnHeading{Classes}
\\
Class parameters
&
\Yes
\\
Inheritance
&
\Yes
\\
\kw{include} statement
&
\Yes
\\
Other multiple-usage mechanisms (\kw{require}, \kw{contain}, \kw{hiera\_include})
&
\No
\\
Resource-like declarations
&
\Yes
\\
\twoColumnHeading{Relationships and ordering}
\\
Relationship metaparameters (\tt{before}, \tt{require}, \tt{notify}, \tt{subscribe})
&
\No
\\
Chaining arrows (\tt{->})
&
\No
\\
\tt{require} function
&
\No
\\
\twoColumnHeading{Advanced constructs}
\\
Plugins (Ruby-level functions)
&
\No
\\
Puppet-level functions
&
\No
\\
Templates
&
\No
\\
Iteration functions (\tt{each}, \tt{slice}, \tt{filter}, \tt{map}, \tt{reduce})
&
\No
\\
Lambdas (\tt{|\ldots| \{\ldots\}})
&
\No
\\
Advanced resource features (defaults, collectors, virtual resources, exported resources)
&
\No
\\
Tags
&
\No
\\
Run stages
&
\No
\end{tabular}
\caption{Summary of Puppet 4.7 language coverage}
\label{fig:language-coverage}
\end{figure}

%% file: rule/environment/lookup.tex
\[\begin{array}{rcll}
\lookup&:&\Set{Scalar}\times \Set{HashValue}\rightarrow \Set{Value} \uplus\{ \bot\}\\
\lookup(k,(k'\Rightarrow v, \HV))&=&
					\lookup(k,\HV)&\text{ if } k\neq k'\\
\lookup(k,(k\Rightarrow v, \HV)) &=&					v\\
\lookup(k,\varepsilon) &=& 					\bot
\end{array}
\]

%% file: rule/environment/lookupCatalog.tex
\[\begin{array}{rcll}
\lookupCat&:&\Set{\CV}\times \Set{String}\times \Set{String}\times \Set{String}\rightarrow \Set{Value} \uplus\{ \bot\}\\
\lookupCat(( \CV\visiblespace t:\{w: \HV\}), t', w', x)&=&
					\lookupCat(\CV, t', w', x) \\
					&&\text{ if } t\neq t' \text{or } w\neq w'  \\
\lookupCat((\CV\visiblespace t:\{w: \HV\}), t, w, x) &=&					\lookup(x,\HV)\\
\lookupCat(\varepsilon,t,w,x) &=& 					\bot
\end{array}
\]

%% file: rule/environment/merge.tex
\[
\begin{array}{rcll}
\merge&: & \Set{Params}\times \Set{HashValue} \rightarrow \Set{Stmt}\\
\merge(\varepsilon,\HV) &=& \kskip\\
\merge((\$x,\rho), \HV) &=& \$x=v\visiblespace\merge(\rho,\HV) & \text{if }
                           v=\lookup(x,\HV)\neq\bot\\
\merge((\$x=e,\rho),\HV) &=& \$x=v\visiblespace\merge(\rho,\HV)& \text{if }v=\lookup(x,\HV)\neq\bot\\
\merge((\$x=e,\rho),\HV) &=& \$x=e\visiblespace\merge(\rho,\HV)& \text{if }
                           \lookup(x,\HV)=\bot
\end{array}
\]

%% file: rule/environment/update.tex
\[\begin{array}{rcll}
\fupdate{-}{-}{-}{-} & :& \Set{Env}\times \Set{Scope}\times \Set{Var}\times \Set{Value}\to \Set{Env}\\
\fupdate{\sigma}{\alpha}{x}{v} &=&\lambda(\alpha', y).\ \text{if }
\alpha=\alpha' \text{ and } x=y \text{ then } v \text{ else }
\sigma(\alpha', y)
\end{array}\]

%% file: rule/environment/specialize.tex
\[\begin{array}{rcll}
-_-&:& \Set{Var} \to \Set{Value}\\
\sigma_{\alpha}&=&\lambda x. \sigma(\alpha,x)
\end{array}\]

%% file: rule/environment/clear.tex
\[\begin{array}{rcl}
\clear&:&\Set{Env}\times \Set{Scope} \to \Set{Env}\\
\clear(\sigma,\alpha) &=&
					\lambda(\alpha',
                          y).\left\{\begin{array}{ll}
\sigma(\alpha', y) & \text{ if } \alpha
                                                  \neq \alpha'\\
\text{undefined} & \text{otherwise}\end{array}\right.

\end{array}\]

%% file: rule/expression/arithmetic.tex
\begin{mathpar}
  \inferrule*[right=ArithLeft]
  {
    \sigma,\kappa, \CV,e_{1}\stepExpr{\alpha} e_{1}'
  }
  {
    \sigma,\kappa, \CV,e_{1}+e_{2}\stepExpr{\alpha} e_{1}'+e_{2}
  }
  \and
  \inferrule*[right=ArithRight]
  {
    \sigma,\kappa,\CV, e\stepExpr{\alpha} e'
  }
  {
    \sigma,\kappa, \CV,i+e\stepExpr{\alpha} i+e'
  }
  \and
  \inferrule*[right=ArithValue]
  {
  }
  {
    \sigma,\kappa, \CV,i_{1}+i_{2}\stepExpr{\alpha}  i_1 +_{\mathbb{Z}} i_2
  }
\end{mathpar}

%% file: rule/expression/comparison.tex
\begin{mathpar}
  \inferrule*[right=CompLeft]
  {
    \sigma,\kappa, \CV,e_{1}\stepExpr{\alpha} e_{1}'
  }
  {
    \sigma,\kappa, \CV,e_{1}> e_{2}\stepExpr{\alpha} e_{1}'> e_{2}
  }
  \and
  \inferrule*[right=CompRight]
  {
    \sigma,\kappa, \CV,e\stepExpr{\alpha} e'
  }
  {
    \sigma,\kappa,\CV, v > e\stepExpr{\alpha} v > e'
  }
  \\
  \inferrule*[right=CompValueI]
  {
    v_{1} >_{\mathbb{Z}} v_{2}
  }
  {
    \sigma,\kappa, \CV,v_{1}> v_{2}\stepExpr{\alpha} \ktrue
  }
  \and
  \inferrule*[right=CompValueII]
  {
    v_{1} \leq_{\mathbb{Z}} v_{2}
  }
  {
    \sigma,\kappa, \CV,v_{1}> v_{2}\stepExpr{\alpha} \kfalse
  }
\end{mathpar}

%% file: rule/expression/boolean.tex
\begin{mathpar}
  \inferrule*[right=AndLeft]
  {
    \sigma,\kappa, \CV,e_{1}\stepExpr{\alpha} e_{1}'
  }
  {
    \sigma,\kappa, \CV,e_{1} \kand e_{2}\stepExpr{\alpha} e_{1}'\kand e_{2}
  }
  \and
  \inferrule*[right=AndRightI]
  {
    \sigma,\kappa, \CV,e \stepExpr{\alpha} e'
  }
  {
    \sigma,\kappa, \CV,\kfalse \kand e \stepExpr{\alpha} \kfalse
  }
  \and
  \inferrule*[right=AndRightII]
  {
    \sigma,\kappa, \CV,e \stepExpr{\alpha} e'
  }
  {
    \sigma,\kappa, \CV,\ktrue \kand e \stepExpr{\alpha} \ktrue \kand e'
  }
  \and
  \inferrule*[right=AndValue]
  {
    \strut
  }
  {
    \sigma, \kappa, \CV,\ktrue \kand \ktrue \stepExpr{\alpha} \ktrue
  }
  \and
  \inferrule*[right=NotStep]
  {
    \sigma,\kappa, \CV,e\stepExpr{\alpha} e'
  }
  {
    \sigma,\kappa, \CV,\Not{e} \stepExpr{\alpha} \Not{e'}
  }
  \and
  \inferrule*[right=NotValueI]
  {
  }
  {
    \sigma,\kappa, \CV,\Not{\ktrue} \stepExpr{\alpha} \kfalse
  }
  \and
  \inferrule*[right=NotValueII]
  {
  }
  {
    \sigma,\kappa, \CV,\Not{\kfalse} \stepExpr{\alpha} \ktrue
  }
\end{mathpar}

%% file: rule/expression/array.tex
\begin{mathpar}
\inferrule*[right=ArrExp]
 {
    \sigma, \kappa, \CV, A\xrightarrow{\alpha} A'
  }
  {
    \sigma,\kappa, \CV, [A] \xrightarrow{\alpha} [A']
  }
  \and
  \inferrule*[right=ArrEleI]
 {
    \sigma, \kappa, \CV, A \xrightarrow{\alpha} A'
  }
  {
    \sigma,\kappa, \CV, v, A \xrightarrow{\alpha} v, A'
  }
  \and
 \inferrule*[right=ArrEleII]
 {
    \sigma, \kappa, \CV, e_{1} \xrightarrow{\alpha} e_{1}'
  }
  {
    \sigma,\kappa, \CV, e_{1},A \xrightarrow{\alpha} e_{1}', A
  }
\end{mathpar}

%% file: rule/expression/hash.tex
\begin{mathpar}
\inferrule*[right=HaExp]
 {
    \sigma,\kappa, \CV, H\xrightarrow{\alpha} H'
  }
  {
    \sigma,\kappa,\CV,  \{H\} \xrightarrow{\alpha} \{H\}
  }
  \and
  \inferrule*[right=HEleI]
 {
    \sigma,\kappa, \CV, H\xrightarrow{\alpha} H'
  }
  {
    \sigma,\kappa,\CV,  \mathtt{k}\Rightarrow v , H\xrightarrow{\alpha} \mathtt{k}\Rightarrow v, H
  }
  \and
 \inferrule*[right=HEleII]
 {
   \sigma,\kappa, \CV, e\xrightarrow{\alpha} e'
  }
  {
    \sigma,\kappa, \CV, \mathtt{k} \Rightarrow e, H\xrightarrow{\alpha} \mathtt{k} \Rightarrow e', H
  }
 \end{mathpar}

%% file: rule/expression/selector.tex
\begin{mathpar}
  \inferrule*[right=SControl]
  {
    \sigma,\kappa , \CV,e \stepExpr{\alpha} e'
  }
  {
    \sigma,\kappa , \CV,e \;\query\{M\} \stepExpr{\alpha} e' \;\query\{M\}
  }
  \and
  \inferrule*[right=SCase]
  {
    \sigma,\kappa, \CV,e_1 \stepExpr{\alpha} e_1'
  }
  {
    \sigma,\kappa, \CV,v\;\query\{e_{1}\Rightarrow e, M\}\stepExpr{\alpha} v\;\query\{e_{1}'\Rightarrow e, M\}
  }
  \and
  \inferrule*[right=SChooseI]
  {
    \caseMatch(v,v_1)
  }
  {
    \sigma,\kappa, \CV,v\;\query\{v_{1}\Rightarrow e, M\}\stepExpr{\alpha} e
  }
  \and
  \inferrule*[right=SChooseII]
  {
    \neg\caseMatch(v,v_{1})
  }
  {
    \sigma,\kappa, \CV,v\;\query\{v_{1}\Rightarrow e, M\}\stepExpr{\alpha} v\;\query\{M\}
  }
  \and
  \inferrule*[right=SDefault]
  {
  }
  {
    \sigma,\kappa, \CV,v\;\query\{\kdefault\Rightarrow e, M  \}\stepExpr{\alpha} e
  }
\end{mathpar}

%% file: rule/resource.tex
\begin{mathpar}
  \inferrule*[right=ResStepII]
  {
    \sigma,\kappa, \CV,e\stepExpr{\alpha} e'
  }
  {
    \sigma,\kappa, \CV,x \Rightarrow e, H\stepHash{\alpha} x \Rightarrow e', H
  }
  \and
  \inferrule*[right=ResStepIII]
  {
    \sigma,\kappa,\CV, H\stepHash{\alpha} H'
  }
  {
    \sigma,\kappa, \CV,x \Rightarrow v , H\stepHash{\alpha} x \Rightarrow v, H'
  }
  \\
  \and
  \inferrule*[right=ResTitle]
  {
    \sigma,\kappa, \CV,e \stepExpr{\alpha}e'
  }
  {
    \sigma,\kappa, \CV,e:H\stepRes{\alpha} e':H
  }
  \and
  \inferrule*[right=ResStepI]
  {
    \sigma,\kappa, \CV,H\stepHash{\alpha}H'
  }
  {
    \sigma,\kappa, \CV,e:H\stepRes{\alpha} e:H'
  }
\end{mathpar}

%% file: rule/statement/expression.tex
\begin{mathpar}
  \inferrule*[right=ExprStep]
  {
\sigma,\kappa,\CV , e \stepExpr{\alpha} e'
  }
  {
    \sigma,\kappa,\CV , e \stepStmt{\alpha} \sigma,\kappa,\CV, e'
  }
\and
  \inferrule*[right=Expr]
  {
  }
  {
    \sigma,\kappa,\CV , v \stepStmt{\alpha} \sigma,\kappa,\CV, \kskip
  }
\end{mathpar}

%% file: rule/statement/sequential-composition.tex
\begin{mathpar}
  \inferrule*[right=SeqStep]
  {
    \sigma,\kappa,\CV, s_1 \stepStmt{\alpha} \sigma',\kappa',\CV' , s_1'
  }
  {
    \sigma,\kappa,\CV, s_1\visiblespace s_2
    \stepStmt{\alpha}
    \sigma',\kappa',\CV' , s_1'\visiblespace s_2
  }
  \and
  \inferrule*[right=SeqSkip]
  {
  }
  {
    \sigma,\kappa,\CV, \kskip\visiblespace s
    \stepStmt{\alpha}
    \sigma,\kappa, \CV, s
  }
\end{mathpar}

%% file: rule/statement/if.tex
\begin{mathpar}
  \inferrule*[right=IfStep]
  {
    \sigma,\kappa , \CV, e\stepExpr{\alpha} e'
  }
  {
    \sigma,\kappa, \CV, \kif\ e\;\{s_1\}  \ \kelse\ \{s_2\}\stepStmt{\alpha}\sigma,\kappa,\CV ,\text{\kif\
  }e'\;\{s_1\} \ \kelse\ \{s_2\}
  }
  \and
  \inferrule*[right=IfT]
  {
  }
  {
    \sigma,\kappa,\CV , \kif\; \ktrue\ \{s_1\} \ \kelse\ \{s_2\}\stepStmt{\alpha} \sigma,\kappa,\CV , s_1
  }
  \and
  \inferrule*[right=IfF]
  {
  }
  {
    \sigma,\kappa,\CV , \kif\; \kfalse\ \{s_1\} \ \kelse\ \{s_2\}\stepStmt{\alpha} \sigma,\kappa,\CV , s_2
  }
\end{mathpar}

%% file: rule/statement/unless.tex
\begin{mathpar}
  \inferrule*[right=UnlessStep]
  {
    \sigma,\kappa, \CV, e\stepExpr{\alpha} e'
  }
  {
    \sigma,\kappa, \CV, \mathtt{unless}\; e\; \{s
  \}\stepStmt{\alpha}\sigma,\kappa,\CV , \mathtt{unless}\; e' \;\{s\}
  }
  \and
  \inferrule*[right=UnlessT]
  {
  }
  {
    \sigma,\kappa,\CV, \kunless\; \ktrue \;\{s \}\stepStmt{\alpha} \sigma,\kappa,\CV,\;\mathtt{skip}
  }
  \and
  \inferrule*[right=UnlessF]
  {
  }
  {
    \sigma,\kappa,\CV, \kunless\; \kfalse \;\{s \}\stepStmt{\alpha} \sigma,\kappa,\CV,\;s
  }
\end{mathpar}

%% file: rule/statement/case.tex
\begin{mathpar}
  \inferrule*[right=CaseStep1]
  {
    \sigma,\kappa, \CV, e\stepExpr{\alpha} e'
  }
  {
    \sigma,\kappa,\CV, \mathtt{case}\; e\ \{C \}\stepStmt{\alpha}\sigma,\kappa,\CV, \mathtt{case}\; e'\ \{C\}
  }
  \and
  \inferrule*[right=CaseStep2]
  {
    \sigma,\kappa, \CV, e\stepExpr{\alpha} e'
  }
  {
    \sigma,\kappa,\CV, \mathtt{case}\; v\ \{e:\{s\}\visiblespace C \}\stepStmt{\alpha}\sigma,\kappa,\CV, \mathtt{case}\; v\ \{e':\{s\}\visiblespace C \}
  }
  \and
  \inferrule*[right=CaseMatch]
  {
    \caseMatch(v,v_1)
  }
  {
    \sigma,\kappa,\CV, \mathtt{case}\; v\ \{v_{1}:\{s\}\visiblespace C \}\stepStmt{\alpha}\sigma,\kappa,\CV,
  s
  }
  \and
  \inferrule*[right=CaseNoMatch]
  {
    \neg\caseMatch(v,v_1)
  }
  {
    \sigma,\kappa,\CV, \mathtt{case}\; v\ \{v_{1}:\{s\}\visiblespace C
  \}\stepStmt{\alpha}\sigma,\kappa,\CV, \mathtt{case}\; v\
  \{C \}
  }
  \and
  \inferrule*[right=CaseDone]
  {
  }
  {
    \sigma,\kappa,\CV, \mathtt{case}\; v\ \{ \varepsilon \}\stepStmt{\alpha}\sigma,\kappa,\CV, \kskip
  }
\end{mathpar}

%% file: rule/manifest/sequential-composition.tex
\begin{mathpar}
  \inferrule*[right=MSeqStep]
  {
    \sigma,\kappa,\CV , m_{1} \stepManifest{N}
    \sigma',\kappa',\CV' , m_{1}'
  }
  {
    \sigma,\kappa,\CV , m_{1}\visiblespace m_{2} \stepManifest{N} \sigma',\kappa',\CV'
    , m_{1}'\visiblespace m_{2}
  }
  \and
  \inferrule*[right=MSeqSkip]
  {
  }
  {
    \sigma,\kappa,\CV , \mathtt{skip}\visiblespace m \stepManifest{N} \sigma,\kappa,\CV
    , m
  }
\end{mathpar}